\documentclass[screen]{acmart}

\usepackage{mathtools}
    
    \def\ZZ{\mathsf Z}
    \def\le{\leqslant}
    \def\ge{\geqslant}
    \allowdisplaybreaks
    
    \def\lnMom{\ln\Mom}
    \def\Mom#1#2#3{\big._{#1} \mkern-1mu \big\| #2 \bigr\| \big._{#3}}

\usepackage{amsthm}
    \newtheorem{theorem}{Theorem}
    \newtheorem{lemma}[theorem]{Lemma}
    \newtheorem{corollary}[theorem]{Corollary}
    \newtheorem{proposition}[theorem]{Proposition}

\usepackage{tikz, pgfplots, pgfplotstable}
    \tikzset{
        every picture/.style={
            line cap=round, line join=round, baseline=-0.5ex,
        }
    }
    \pgfplotsset{
        compat=1.18,
        colormap={red--blue}{rgb=(1,0,0); rgb=(1,0,1); rgb=(0,0,1)},
        graph marker/.style={
            scatter, only marks, point meta=explicit symbolic,
            scatter/@pre marker code/.style={/tikz/mark=\pgfplotspointmeta},
            scatter/@post marker code/.style={},
        }
    }
    \def\pushtriangle#1{\xdef\trianglebuffer{\trianglebuffer #1}}
    \def\computexyz&meta(#1,#2,#3)[#4]=\f(#5,#6){
        \pgfmathsetmacro#1{#5}
        \pgfmathsetmacro#2{#6}
        \pgfmathsetmacro#3{\f(#1,#2)}
        \pgfmathsetmacro#4{\g(#3)}
    }
    \def\addplottriangles{
        \preparetriangles
        \edef\pgfmarshal{
            \noexpand\addplot3 [patch, patch type=triangle,
                point meta=explicit] coordinates {\trianglebuffer};
        }
        \pgfmarshal
    }

\begin{document}

                                   \title
                    {Ambidextrous Degree Sequence Bounds
                   for Pessimistic Cardinality Estimation}

                            \author{Yu-Ting Lin}
   \affiliation{\institution{National Taiwan University}\country{Taipei}}
                        \email{r14942055@ntu.edu.tw}
                            \author{Hsin-Po Wang}
   \affiliation{\institution{National Taiwan University}\country{Taipei}}
                          \email{hsinpo@ntu.edu.tw}
                                      
\begin{abstract}
    In a large database system, upper-bounding the cardinality of a join
    query is a crucial task called \emph{pessimistic cardinality estimation}.
    Recently, Abo Khamis, Nakos, Olteanu, and Suciu unified related works
    into the following dexterous framework. Step 1: Let $(X_1, \dotsc, X_n)$
    be a random row of the join, equating $H(X_1, \dotsc, X_n)$ to the log of
    the join cardinality.  Step 2: Upper-bound $H(X_1, \dotsc, X_n)$ using
    Shannon-type inequalities such as $H(X, Y, Z) \le H(X) + H(Y|X) +
    H(Z|Y)$.  Step 3: Upper-bound $H(X_i) + p H(X_j | X_i)$ using the
    $p$-norm of the degree sequence of the underlying graph of a relation.

    While old bound in step 3 count ``claws $\in$'' in the underlying graph,
    we proposed \emph{ambidextrous} bounds that count ``claw pairs
    ${\ni}\!{-}\!{\in}$''.  The new bounds are provably not looser and
    empirically tighter: they overestimate by $x^{3/4}$ times when the old
    bounds overestimate by $x$ times.  An example is counting friend triples
    in the \texttt{com-Youtube} dataset, the best dexterous bound is $1.2
    \cdot 10^9$, the best ambidextrous bound is $5.1 \cdot 10^8$, and the
    actual cardinality is $1.8 \cdot 10^7$.
\end{abstract}

\maketitle

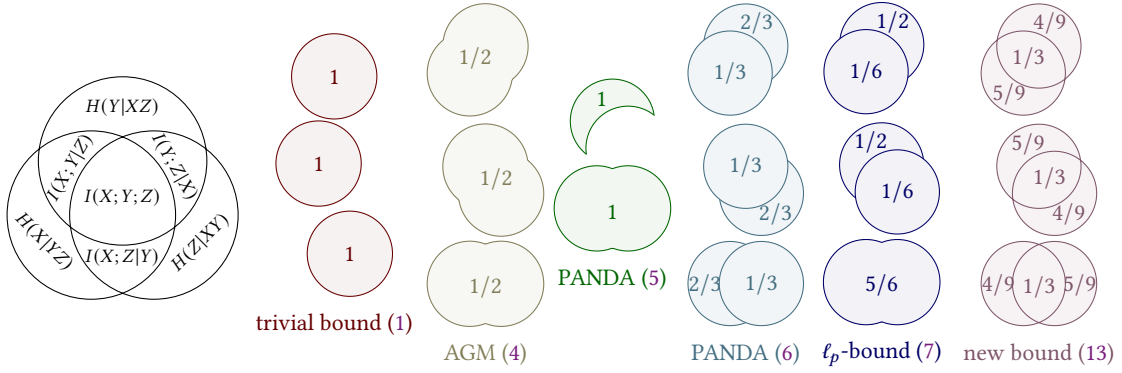
\begin{figure}
    \centering
    \tikzset{every picture/.append style={baseline=(current bounding box)}}
    \begin{tikzpicture} [scale=0.8, every node/.style={scale=0.8}]
        \draw
            (90: 0.6) circle (1.4)
            (210: 0.6) circle (1.4)
            (-30: 0.6) circle (1.4)
            (0, 0) node {$I(X; Y; Z)$}
            (90: 1.5) node {$H(Y | XZ)$}
            (210: 1.5) node [rotate=-45] {$H(X | YZ)$}
            (-30: 1.5) node [rotate=45] {$H(Z | XY)$}
            (150: 1) node [rotate=60] {$I(X; Y | Z)$}
            (30: 1) node [rotate=-60] {$I(Y; Z | X)$}
            (-90: 1) node {$I(X; Z | Y)$}
        ;
    \end{tikzpicture}
    \begin{tikzpicture} [scale=0.8, red!40!black]
        \draw [line width=0.8pt, postaction={fill=.!5}, yshift=1cm]
            (90: 0.3) circle (0.7) node {$1$};
        \draw [line width=0.8pt, postaction={fill=.!5}, yshift=0cm]
            (210: 0.3) circle (0.7) node {$1$};
        \draw [line width=0.8pt, postaction={fill=.!5}, yshift=-1.5cm]
            (-30: 0.3) circle (0.7) node {$1$};
        \draw (0, -2.5) node [below] {trivial bound \eqref{triple}};
    \end{tikzpicture}
    \hfill
    \begin{tikzpicture} [scale=0.8, yellow!40!black]
        \draw [line width=0.8pt, postaction={fill=.!5}, yshift=2cm]
            (210: 0.3) circle (0.7) (90: 0.3) circle (0.7)
            (150: 0.2) node {$1/2$};
        \draw [line width=0.8pt, postaction={fill=.!5}, yshift=0cm]
            (90: 0.3) circle (0.7) (-30: 0.3) circle (0.7)
            (30: 0.2) node {$1/2$};
        \draw [line width=0.8pt, postaction={fill=.!5}, yshift=-1.5cm]
            (210: 0.3) circle (0.7) (-30: 0.3) circle (0.7)
            (270: 0.2) node {$1/2$};
        \draw (0, -2.5) node [below] {AGM \eqref{half}};
    \end{tikzpicture}
    \hfill
    \begin{tikzpicture} [scale=0.8, green!40!black]
        \begin{scope} [yshift=-0.5cm]
            \draw [fill=.!5] (90: 0.3) circle (0.7);
            \fill [white] (-30: 0.3) circle (0.7);
            \clip (90: 0.3) circle (0.7);
            \draw (-30: 0.3) circle (0.7) (105: 0.7) node {$1$};
        \end{scope}
        \draw [line width=0.8pt, postaction={fill=.!5}, yshift=-1.5cm]
            (210: 0.3) circle (0.7) (-30: 0.3) circle (0.7)
            (270: 0.2) node {$1$};
        \draw (0, -2.5) node [below, align=center]
            {PANDA \eqref{PANDA}};
    \end{tikzpicture}
    \hfill
    \begin{tikzpicture} [scale=0.8, cyan!40!black]
        \draw [fill=.!5, yshift=2cm]
            (90: 0.3) circle (0.7) (75: 0.7) node {$2/3$};
        \draw [fill=.!5, yshift=2cm]
            (210: 0.3) circle (0.7) node {$1/3$};
        \draw [fill=.!5, yshift=0cm]
            (-30: 0.3) circle (0.7) (-45: 0.7) node {$2/3$};
        \draw [fill=.!5, yshift=0cm]
            (90: 0.3) circle (0.7) node {$1/3$};
        \draw [fill=.!5, yshift=-1.5cm]
            (210: 0.3) circle (0.7) (195: 0.7) node {$2/3$};
        \draw [fill=.!5, yshift=-1.5cm]
            (-30: 0.3) circle (0.7) node {$1/3$};
        \draw (0, -2.5) node [below] {PANDA \eqref{222}};
    \end{tikzpicture}
    \hfill
    \begin{tikzpicture} [scale=0.8, blue!40!black]
        \draw [fill=.!5, yshift=2cm]
            (90: 0.3) circle (0.7) (75: 0.7) node {$1/2$};
        \draw [fill=.!5, yshift=2cm]
            (210: 0.3) circle (0.7) node {$1/6$};
        \draw [fill=.!5, yshift=0cm]
            (90: 0.3) circle (0.7) (105: 0.7) node {$1/2$};
        \draw [fill=.!5, yshift=0cm]
            (-30: 0.3) circle (0.7) node {$1/6$};
        \draw [fill=.!5, yshift=-1.5cm]
            (210: 0.3) circle (0.7) (195: 0.7) node {$2/3$};
        \draw [line width=0.8pt, postaction={fill=.!5}, yshift=-1.5cm]
            (210: 0.3) circle (0.7) (-30: 0.3) circle (0.7)
            (270: 0.2) node {$5/6$};
        \draw (0, -2.5) node [below] {$\ell_p$-bound \eqref{335}};
    \end{tikzpicture}
    \hfill
    \begin{tikzpicture} [scale=0.8, magenta!40!black]
        \draw [fill=.!5, yshift=2cm]
            (210: 0.3) circle (0.7) (90: 0.3) circle (0.7)
            (75: 0.7) node {$4/9$}
            (150: 0.2) node {$1/3$}
            (225: 0.7) node {$5/9$};
        \draw [fill=.!5, yshift=0cm]
            (90: 0.3) circle (0.7) (-30: 0.3) circle (0.7)
            (-45: 0.7) node {$4/9$}
            (30: 0.2) node {$1/3$}
            (105: 0.7) node {$5/9$};
        \draw [fill=.!5, yshift=-1.5cm]
            (-30: 0.3) circle (0.7) (210: 0.3) circle (0.7)
            (195: 0.7) node {$4/9$}
            (-90: 0.2) node {$1/3$}
            (-15: 0.7) node {$5/9$};
        \draw (0, -2.5) node [below] {new bound \eqref{345}};
    \end{tikzpicture}
    \tikzset{every picture/.append style={baseline=-0.5ex}}
    \newbox\vennbox\setbox\vennbox=\hbox{%
        \bfseries
        \begin{tikzpicture} [scale=0.15]
        \draw (0, 0) circle (1);
        \end{tikzpicture},
        \begin{tikzpicture} [scale=0.15]
            \draw (0, 0) circle (1) (1, 0) circle (1);
            \fill [white] (0, 0) circle (1cm-1pt) (1, 0) circle (1cm-1pt);
        \end{tikzpicture},
        \begin{tikzpicture} [scale=0.15]
            \draw (0, 0) circle (1);
            \fill [white] (1, 0) circle (1);
            \clip (0., 0) circle (1);
            \draw (1, 0) circle (1);
        \end{tikzpicture}\kern-0.2cm, and
        \begin{tikzpicture} [scale=0.15]
            \draw (0, 0) circle (1);
            \draw [fill=white] (1, 0) circle (1);
        \end{tikzpicture}%
    }
    \newbox\newblock\setbox\newblock=\hbox{%
        \begin{tikzpicture} [scale=0.15]
            \draw (0, 0) circle (1) (1, 0) circle (1);
        \end{tikzpicture}
    }
    \Description{
        A visualization of \eqref{triple}, \eqref{half}, \eqref{PANDA},
        \eqref{222}, \eqref{335}, and \eqref{345} as coverings of the entropy
        Venn diagram.
    }
    \caption{
        Bounding $\#_\Delta$ can be viewed as a fractional covering problem
        on the entropy Venn diagram [left most].  The building blocks are
        \usebox\vennbox, which correspond to \eqref{p=0}, \eqref{p=1},
        \eqref{p=oo}, and \eqref{p=p}, respectively.  Our contribution can be
        viewed as inventing a new building block \usebox\newblock, which
        corresponds to \eqref{p1q}.  This generates new bounds such as
        \eqref{345}.  See Lemma~\ref{lem:venn} to learn how to generate more.
        Note that this Venn diagram viewpoint is not valid for four or more
        variables, as the signs of intersection information terms are not
        clear.
    }                                                        \label{fig:venn}
\end{figure}

\section{Introduction}

    In a large database, joining multiple relations is a basic operation, yet
    it could be costly to execute.  Often a query optimizer wants to know if
    one execution plan, say $(R \land S) \land T$, is better than the other,
    say $R \land (S \land T)$.  Therefore, estimating the costs of different
    plans, primarily the cardinalities of intermediate relations, helps
    planning and resource allocation \cite{PKB20,LBP21,KPN22}.  Many works
    had invented and studied various techniques to estimate the cardinality
    of a join, such as sampling and sketching.  See \cite{LGM15} for the
    survey that proposed the join order benchmark (JOB), \cite{HWW21} for the
    STATS benchmark, and \cite{SuL20} for subgraph matching.  Machine
    learning \cite{SZS21,WQW21,KJS22,DZZ24} is also a promising direction,
    though it is not a universal panacea.

    The challenge is that an estimate can undershoot or overshoot by an
    arbitrary amount \cite{LDN23}.  So naturally we seek one-sided bounds on
    the join cardinality that are tight and cheap.  The lower-bound side is
    called \emph{optimistic cardinality estimation} and the upper-bound side
    \emph{pessimistic}.  Together they sanity-check each other \cite{CHW22}.
    This paper works on the pessimistic side.  Readers are referred to
    \cite{ADO25,KNO25} for surveys and \cite{AHS25} (utilizes fast matrix
    multiplication), \cite{DAB25} (planning for tensor computations)
    \cite{IMN25} (polynomial-time bound-optimizer), \cite{UoL25} (generalizes
    entropic techniques to quantum), and \cite{WuS25} (parallel
    implementation of query execution) for very recent works not yet
    surveyed.
    
    To demonstrate progress of pessimistic cardinality estimation, let us
    consider counting triangles among three binary relations
    \[
        \#_\Delta \coloneqq |R(A, B) \land S(B, C) \land T(C, A)|
        = \left| \left\{ (a, b, c) :
            \substack{
                R(a, b) \text{ and} \\
                S(b, c) \text{ and} \\
                T(c, a) \phantom{\text{ and}}
            }
        \right\} \right|
    \]
    as an example.  The starting point is that the number of triangles is at
    most the number of free triples
    \begin{equation}
        \#_\Delta \le |A| \cdot |B| \cdot |C|.                 \label{triple}
    \end{equation}
    This could overshoot by orders of magnitude if the relations $R$, $S$,
    and $T$ are sparse.
    
    To harvest sparsity, one notices that a triangle is determined by two
    edges.  Hence,
    \begin{equation}
        \#_\Delta \le
        |R(A, B)| |S(B, C)|,
        |S(B, C)| |T(C, A)|,
        |T(C, A)| |R(A, B)|.                                    \label{2edge}
    \end{equation}
    Utilizing Shearer's inequality, Atserias, Grohe, and Marx
    \cite{GrM14,AGM13} generalized \eqref{2edge} (see also \cite{FrK98} for
    its math counterpart).  The new estimates read
    \begin{equation}
        \#_\Delta \le
        |R(A, B)|^u \cdot |S(B, C)|^v \cdot |T(C, A)|^w           \label{AGM}
    \end{equation}
    where the weights (say $u$ and $v$) associated to an attribute (say $B$)
    sum to $1$ or more.  In particular, $(u, v, w) = (0, 1, 1)$ gives back
    the second of \eqref{2edge} and $(u, v, w) = (1/2, 1/2, 1/2)$ gives a
    nontrivial one
    \begin{equation}
        \#_\Delta \le |R(A, B)|^{1/2} \cdot
        |S(B, C)|^{1/2} \cdot |T(C, A)|^{1/2}.                   \label{half}
    \end{equation}
    These are referred to as the \emph{AGM bounds}.

    Despite good estimates like \eqref{AGM}, we might still have to run
    nested for-loops in one of the following ways to compute the join $R(A,
    B) \land S(B, C) \land T(C, A)$:
    \begin{itemize}
        \item For $a \in A$ \{ for $b \in B$ s.t.\ $R(a, b)$
            \{ for $c \in C$ s.t.\ ... \{ ... \} \} \}
        \item For $b \in B$ \{ for $c \in C$ s.t.\ $S(b, c)$
            \{ for $a \in A$ s.t.\ ... \{ ... \} \} \}
        \item For $c \in C$ \{ for $a \in A$ s.t.\ $T(c, a)$
            \{ for $b \in B$ s.t.\ ... \{ ... \} \} \}
        \item For $a \in A$ \{ for $c \in C$ s.t.\ $T(c, a)$
            \{ for $b \in B$ s.t.\ ... \{ ... \} \} \}
        \item For $b \in B$ \{ for $a \in A$ s.t.\ $R(a, b)$
            \{ for $c \in C$ s.t.\ ... \{ ... \} \} \}
        \item For $c \in C$ \{ for $b \in B$ s.t.\ $S(b, c)$
            \{ for $a \in A$ s.t.\ ... \{ ... \} \} \}
    \end{itemize}
    And it is likely that all six choices require more time than what is
    estimated, and thus more time than the size of the final table.  What
    makes AGM bounds so profound is that this is not the case.  It is
    possible to design an informed enumeration plan to match the least upper
    bound over the choices of $(u, v, w)$.  See NPRR \cite{NRR14}, LFTJ
    \cite{Vel14}, and Generic-join \cite{NPR18} for more information on the
    bound-matching track.

    The next breakthrough is \emph{chain bounds} \cite{ANS16} and \emph{PANDA
    bounds} \cite{ANS17}.  They harvest the fact that the two edges must
    share the same vertex, and hence choosing the second edge is constrained
    by the degree of the shared vertex.  This leads to
    \begin{align}
        \#_\Delta \le
        |R(A, B)| \cdot \max_{a\in A} \deg_T(a)\;,\;
        |R(A, B)| \cdot \max_{b\in B} \deg_S(b)\;,\;                 \notag\\
        |S(B, C)| \cdot \max_{b\in B} \deg_R(b)\;,\;
        |S(B, C)| \cdot \max_{c\in C} \deg_T(c)\;,\;                 \notag\\
        |T(C, A)| \cdot \max_{c\in C} \deg_S(c)\;,\;
        |T(C, A)| \cdot \max_{a\in A} \deg_R(a)\;.\;            \label{PANDA}
    \end{align}
    Here, $\deg_R(a)$ is the number of $b \in B$ that satisfy $R(a, b)$, and
    other degrees are defined similarly.

    Practical implementation of said bounds include MOLP \cite{CHW22} and
    \textsc{BoundSketch} \cite{CBS19,HHH21}.  Both \cite{ANS16} and
    \cite{ANS17} also showed that some joining algorithms use the time
    predicted by the best bound.  See \cite{GLV12,GoT17,IMN25} for special
    treatments when $\max \deg_R(a) = 1$.  This is called a \emph{functional
    dependency} as such a relation can be seen as a partial function: $R(a)$
    is either $b$ (if $b \in B$ is the only possibility such that $R(a, b)$)
    or undefined.

    An apparent downside of using max-degree is that a relation might have a
    few unrepresentative vertices with very high degrees.  But in real-world
    recommendation-related applications, we are only interested in the top
    $10$ restaurants, movies, and merchandises.  This means that the
    max-degree can be \emph{capped} at $10$ \cite{AFP09,ACK11,ALK13}, making
    estimates like \eqref{PANDA} very useful.

\subsection{The unification of bounds}

    Abo Khamis, Nakos, Olteanu, and Suciu \cite{ANO24} made huge progress by
    generalizing \eqref{triple}--\eqref{PANDA} to an infinite family of
    bounds called $\ell_p$-bounds.  Two new instances include
    \begin{equation}
        \#_\Delta^3 \le \sum_{a\in A}\deg_R(a)^2 \cdot
        \sum_{b\in B}\deg_S(b)^2 \cdot \sum_{c\in C}\deg_T(c)^2   \label{222}
    \end{equation}
    and
    \begin{equation}
        \#_\Delta^6 \le \sum_{a\in A}\deg_R(a)^3 \cdot
        \sum_{c\in C}\deg_S(c)^3 \cdot |T(C, A)|^5                \label{335}
    \end{equation}
    Now the problem boils down to enumerating this rich family of bounds and
    returning a competitively small, if not the smallest, one.  An algorithm
    called \textsc{LpBound} implemented this and is described in the work of
    Zhang et al.\ \cite{ZMA25}.  A similar variant is called
    \textsc{SafeBound} and is described by Deeds et al.\
    \cite{DSB22,DSB23,DSB25}.  See also the website of Mayer et al.
    \cite{MZA25} for live demo.

    Note that the right-hand sides of \eqref{triple}--\eqref{335} all have
    one thing in common: $|A| = \sum_{a\in A} 1$ is the $0$-norm, $|R(A, B)|
    = \sum_{a\in A} \deg_R(a)$ is the $1$-norm, $\sum_{a\in A} \deg_R(a)^2$
    and $\sum_{a\in A} \deg_R(a)^3$ are the $2$-norm squared and the $3$-norm
    cubed, and $\max_{a\in A} \deg_R(a)$ is the $\infty$-norm of the
    \emph{degree sequence} $\{\deg_R(a)\}_{a\in A}$.  Abo Khamis et al.\
    \cite{ANO24} unified all these bounds using a single building block.

    The meta reason degree sequences are useful here is their relation to
    join cardinalities via entropy inequalities.  For instance, for any
    random pair $(X, Y) \in R(A, B)$,
    \begin{align}
        H(X) & \le \ln |A|,                                       \label{p=0}
        \\ H(X, Y) & \le \ln |R(A, B)|,                           \label{p=1}
        \\ H(Y|X) & \le \ln \max_{a\in A} \deg_R(a).             \label{p=oo}
    \end{align}
    Abo Khamis et al.\ \cite{ANO24} unified them as
    \begin{equation}
        H(X) + p H(Y|X) \le \ln \sum_{a\in A} \deg_R(a)^p         \label{p=p}
    \end{equation}
    for any $p \ge 0$, making \eqref{p=0}, \eqref{p=1}, and \eqref{p=oo}
    special cases at $p = 0$, $1$, and $\infty$, respectively.  See
    Figure~\ref{fig:venn} for a visualization of \eqref{triple}--\eqref{335}
    as fractional coverings of the entropy Venn diagram using the left-hand
    sides of \eqref{p=0}--\eqref{p=p}.

    When $p$ happens to be an integer, say $p = 3$, \eqref{p=p} has a
    combinatorial interpretation: The left-hand side is the entropy of a
    random quadruple $(X, Y_1, Y_2, Y_3)$ such that the marginal distribution
    of each $(X, Y_i)$ looks like a copy of $(X, Y)$.  The right-hand side of
    \eqref{p=p} is the number of quadruples $(a, b_1, b_2, b_3) \in A \times
    B^3$ such that each $(a, b_i)$ is in $R(A, B)$.  In graph theory, the
    graph where one vertex connects to the other three, ``$\in$'', is called
    a \emph{claw}.  So \eqref{p=p} translates into that the entropy of a
    random claw is at most the logarithm of the number of claws.  See
    Figure~\ref{fig:311} for enumeration of claws in a Z-shaped relation $\ZZ
    \coloneqq \{ (1, 2), (3, 2), (3, 4) \}$.
    See Figure~\ref{fig:113} for enumeration of claws in the same relation
    $\ZZ$ with the roots on the other side.
\begin{tikzpicture} [overlay]
    \path
        (0, 1) + (112.5:0.1) coordinate (1a)
        (0, 1) + (112.5:0.0) coordinate (1b)
        (0, 1) + (112.5:-.1) coordinate (1c)
        (0, 0) + (112.5:0.1) coordinate (3a)
        (0, 0) + (112.5:0.0) coordinate (3b)
        (0, 0) + (112.5:-.1) coordinate (3c)
        (1, 1) + (112.5:0.1) coordinate (2a)
        (1, 1) + (112.5:0.0) coordinate (2b)
        (1, 1) + (112.5:-.1) coordinate (2c)
        (1, 0) + (112.5:0.1) coordinate (4a)
        (1, 0) + (112.5:0.0) coordinate (4b)
        (1, 0) + (112.5:-.1) coordinate (4c)
    ;
    \path [save path=\zpath]
        (1b) circle (4pt) -- (2b) circle (4pt) --
        (3b) circle (4pt) -- (4b) circle (4pt);
    \gdef\Zpath{\zpath}
\end{tikzpicture}

\begin{figure}
    \centering
    \def\claw#1#2#3#4;{%
        \begin{tikzpicture}
            \draw [use path=\Zpath, line width=4pt, draw=yellow!60!white];
            \draw [red!20!blue] (#1b) -- (#2a);
            \draw [red!50!blue] (#1b) -- (#3b);
            \draw [red!80!blue] (#1b) -- (#4c);
            \draw (0.5, -0.1) node [below] {$(#1, #2, #3, #4)$};
            \draw (-0.3, -0.3) (1.3, 1.3);
        \end{tikzpicture}%
    }%
    \claw 1 222;
    \claw 3 222;
    \claw 3 224;
    \claw 3 242;
    \claw 3 244;
    \claw 3 422;
    \claw 3 424;
    \claw 3 442;
    \claw 3 444;
    \Description{
        Nine quadruples that are compatible with the Z-shaped relation.
    }
    \caption{
        The number of quadruples $(a, b_1, b_2, b_3)$ such that $(a, b_1),
        (a, b_2), (a, b_3) \in \ZZ \coloneqq \{ (1, 2), (3, 2), (3, 4) \}$ is
        $\Mom{3}{\ZZ}{1} = 9$.
    }                                                         \label{fig:311}
\end{figure}
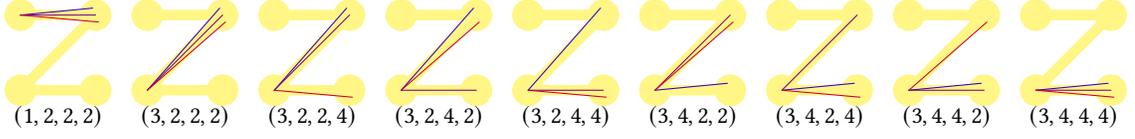

\begin{figure}
    \centering
    \def\claw#1#2#3#4;{%
        \begin{tikzpicture}
            \draw [use path=\Zpath, line width=4pt, draw=yellow!60!white];
            \draw [red!20!blue] (#1a) -- (#4b);
            \draw [red!50!blue] (#2b) -- (#4b);
            \draw [red!80!blue] (#3c) -- (#4b);
            \draw (0.5, -0.1) node [below] {$(#1, #2, #3, #4)$};
            \draw (-0.3, -0.3) (1.3, 1.3);
        \end{tikzpicture}%
    }%
    \claw 111 2;
    \claw 113 2;
    \claw 131 2;
    \claw 133 2;
    \claw 311 2;
    \claw 313 2;
    \claw 331 2;
    \claw 333 2;
    \claw 333 4;
    \Description{
        Nine quadruples that are compatible with the Z-shaped relation.
    }
    \caption{
        The number of quadruples $(a_1, a_2, a_3, b)$ such that $(a_1, b),
        (a_2, b), (a_3, b) \in \ZZ \coloneqq \{ (1, 2), (3, 2), (3, 4) \}$ is
        $\Mom{1}{\ZZ}{3} = 9$.
    }                                                         \label{fig:113}
\end{figure}
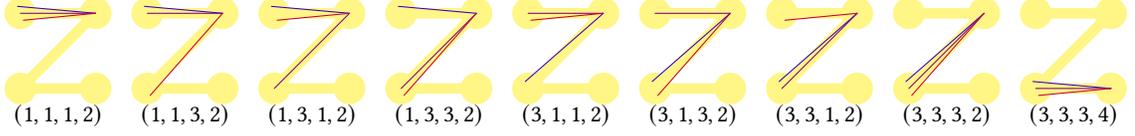

\subsection{Our contributions}

    The main contribution of this paper is to introduce \emph{bivariate
    moments}
    \[
        \Mom{p}{R(A, B)}{q} \coloneqq
        \sum_{(a,b)\in R} \deg_R(a)^{p-1} \deg_R(b)^{q-1}
    \]
    of a \emph{bi-degree sequence}
    $\{(\deg_R(a), \deg_R(b))\}_{(a,b)\in R(A,B)}$.
    Bivariate moments recover the old $p$-norms when $q = 1$:
    \begin{align*}
        & \Mom{0}{R(A, B)}{1}
        = \sum_{(a,b)\in R} \deg_R(a)^{-1}
        = \sum_{a\in A} 1 = |A|                     \tag{used in \eqref{p=0}}
        \\ & \Mom{1}{R(A, B)}{1}
        = \sum_{(a,b)\in R} 1 = |R(A, B)|           \tag{used in \eqref{p=1}}
        \\ & \Mom{p}{R(A, B)}{1}
        = \sum_{(a,b)\in R} \deg_R(a)^{p-1}
        = \sum_{a\in A} \deg_R(a)^p                 \tag{used in \eqref{p=p}}
        \\ & \lim_{p\to\infty} \; \Mom{p}{R(A, B)}{1}^{1/p}
        = \lim_{p\to\infty} \Bigl( \sum_{a\in A} \deg_R(a)^p \Bigr)^{1/p}
        = \max_{a\in A} \deg_A(a)                  \tag{used in \eqref{p=oo}}
    \end{align*}
    We then generalize \eqref{p=0}--\eqref{p=p} to
    \begin{equation}
        pH(Y|X) + I(X; Y) + q H(X|Y) \le \lnMom{p}{R(A, B)}{q} \label{p1q}
    \end{equation}
    for all $p, q \ge 1$, making \eqref{p=p} a special case at $q = 1$.  This
    provides new building blocks to bound $\#_\Delta$, for instance
    \begin{equation}
        \#_\Delta^3 \le
        \Mom{4/3}{R(A, B)}{5/3} \cdot
        \Mom{4/3}{S(B, C)}{5/3} \cdot
        \Mom{4/3}{T(C, A)}{5/3}.                                 \label{345}
    \end{equation}
    See Figure~\ref{fig:venn} for the fractional covering \eqref{345}
    induces.

    When $p$ and $q$ are integers, say $p = q = 2$, \eqref{345} has a
    combinatorial interpretation: The left-hand side is the entropy of a
    random sextuple $(X_1, X_2, Y, X, Y_1, Y_2)$ such that $(X_1, Y)$, $(X_2,
    Y)$, $(X, Y)$, $(X, Y_1)$, and $(X, Y_2)$ all look like copies of $(X,
    Y)$.  The right-hand side of \eqref{p1q} is the number of sextuples
    $(a_1, a_2, b, a, b_1, b_2)$ such that $(a_1, b)$, $(a_2, b)$, $(a, b)$,
    $(a, b_1)$, and $(a, b_2)$ are all in $R(A, B)$.  So \eqref{p1q}
    translates into that the entropy of a random claw pair
    ``${\ni}\!{-}\!{\in}$'' is at most the logarithm of the number of claw
    pairs.  See Figures \ref{fig:111}--\ref{fig:313} for enumerations of
    claw pairs in the Z-shaped relation $\ZZ$.

\begin{figure}
    \centering
    \def\claw#1#2;{%
        \begin{tikzpicture}
            \draw [use path=\Zpath, line width=4pt, draw=yellow!60!white];
            \draw [red!50!blue] (#2b) -- (#1b);
            \draw (0.5, -0.1) node [below] {$(#1, #2)$};
            \draw (-0.3, -0.3) (1.3, 1.3);
        \end{tikzpicture}%
    }%
    \claw 2 1;
    \claw 2 3;
    \claw 4 3;
    \Description{
        Three pairs that are compatible with the Z-shaped relation.
    }
    \caption{
        The number of pairs $(b, a)$ such that $(a, b)\in \ZZ \coloneqq \{
        (1, 2), (3, 2), (3, 4) \}$ is $\Mom{1}{\ZZ}{1} = 3$.
    }                                                         \label{fig:111}
\end{figure}
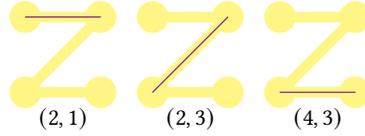

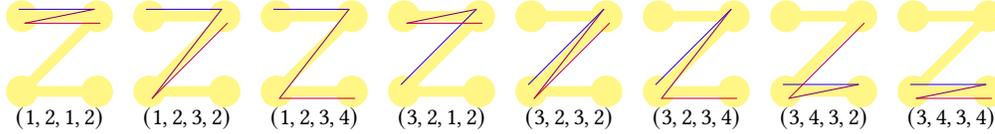
\begin{figure}
    \centering
    \def\claw#1#2#3#4;{%
        \begin{tikzpicture}
            \draw [use path=\Zpath, line width=4pt, draw=yellow!60!white];
            \draw [red!20!blue] (#1a) -- (#2a);
            \draw [red!50!blue] (#3c) -- (#2a);
            \draw [red!80!blue] (#3c) -- (#4c);
            \draw (0.5, -0.1) node [below] {$(#1, #2, #3, #4)$};
            \draw (-0.3, -0.3) (1.3, 1.3);
        \end{tikzpicture}%
    }%
    \claw 1 2 1 2;
    \claw 1 2 3 2;
    \claw 1 2 3 4;
    \claw 3 2 1 2;
    \claw 3 2 3 2;
    \claw 3 2 3 4;
    \claw 3 4 3 2;
    \claw 3 4 3 4;
    \Description{
        Eight quadruples that are compatible with the Z-shaped relation.
    }
    \caption{
        The number of quadruples $(a_1, b, a, b_1)$ such that $(a, b), (a,
        b_1), (a_1, b) \in \ZZ \coloneqq \{ (1, 2), (3, 2), (3, 4) \}$ is
        $\Mom{2}{\ZZ}{2} = 8$.
    }                                                         \label{fig:212}
\end{figure}

    We call \eqref{p=p} \emph{dexterous} for them counting claws
    ``$\in$''.  We call the new bounds \eqref{p1q} \emph{ambidextrous} for
    them counting claw pairs ``${\ni}\!{-}\!{\in}$''.  Ambidextrous upper
    bounds are provably tighter (or at least not looser) in the sense of
    Hölder's inequality
    \begin{equation}
        \Mom{p}{R}{1}^w \cdot \Mom{1}{R}{q}^{1-w}
        \ge \Mom{wp+(1-w)}{R}{w+(1-w)q}.                     \label{holder}
    \end{equation}
    That is to say, if some variant of \eqref{triple}--\eqref{335} uses the
    left-hand side of \eqref{holder} to upper-bound $\#_\Delta$, we can
    instead use the right-hand side of \eqref{holder} and guarantee an
    equally good, if not better, bound.  See Figure~\ref{fig:holder} for a
    graphical demonstration of $\Mom{3}{\ZZ}{1}^{1/2} \cdot
    \Mom{1}{\ZZ}{3}^{1/2} > \Mom{2}{\ZZ}{2}$.  Empirically, estimating
    triangles is now 48\%, quadrilaterals $2.6$-fold, and pentagons
    $4.2$-fold more precise (see Figure~\ref{fig:average} for context).

    As a bonus contribution, we show that finding the best bound, dexterous
    or ambidextrous, is a convex optimization problem.  This is due to the
    following generalization of \eqref{holder}
    \begin{equation}
        \Mom{p}{R}{q}^w \cdot \Mom{r}{R}{s}^{1-w}
        \ge \Mom{wp+(1-w)r}{R}{wq+(1-w)s}.                   \label{holder2}
    \end{equation}
    Convexity simplifies the optimization process in several ways.  For one,
    we can take advantages of existing black-box algorithms.  We can also
    sample the target function at several points and treat the remainder as a
    linear programming problem; this is equivalent to using linear splines
    and is the approach earlier works \cite{DSB22,DSB23,DSB25,ZMA25,MZA25}
    take.  Higher-degree splines approximate convex function even better with
    fewer samples; this is good news as evaluating $\lnMom{p}{R}{q}$ can be
    very I/O intensive.

    The paper is organized as below.  Section~\ref{sec:old} reviews the
    existing bounds and the framework built upon them.  Section~\ref{sec:new}
    explains the new ambidextrous bounds and some consequences.
    Section~\ref{sec:simu} presents simulation results.

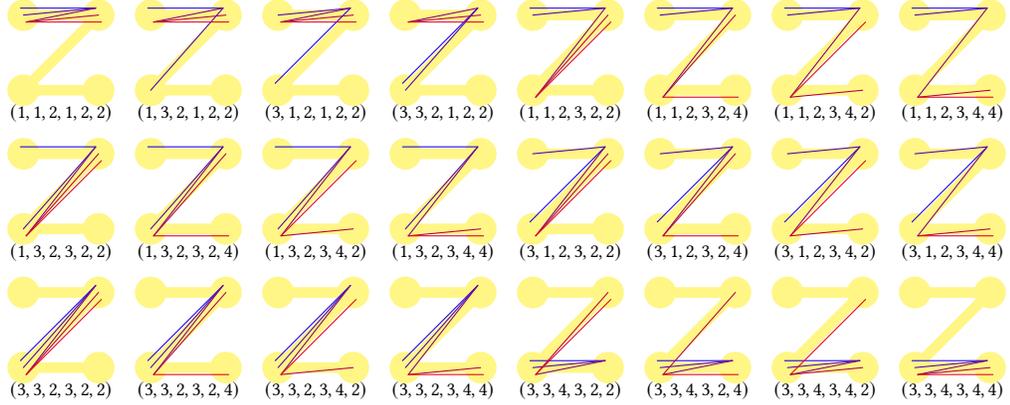
\begin{figure}
    \centering
    \def\claw#1#2#3#4#5#6;{%
        \begin{tikzpicture}
            \draw [use path=\Zpath, line width=4pt, draw=yellow!60!white];
            \draw [red!10!blue] (#1a) -- (#3a);
            \draw [red!30!blue] (#2b) -- (#3a);
            \draw [red!50!blue] (#4c) -- (#3a);
            \draw [red!70!blue] (#4c) -- (#5b);
            \draw [red!90!blue] (#4c) -- (#6c);
            \draw (0.5, -0.1) node [below, scale=0.8]
                {$(#1, #2, #3, #4, #5, #6)$};
            \draw (-0.3, -0.3) (1.3, 1.3);
        \end{tikzpicture}%
    }%
    \claw 11 2 1 22;
    \claw 13 2 1 22;
    \claw 31 2 1 22;
    \claw 33 2 1 22;
    \claw 11 2 3 22;
    \claw 11 2 3 24;
    \claw 11 2 3 42;
    \claw 11 2 3 44;
    \claw 13 2 3 22;
    \claw 13 2 3 24;
    \claw 13 2 3 42;
    \claw 13 2 3 44;
    \claw 31 2 3 22;
    \claw 31 2 3 24;
    \claw 31 2 3 42; 
    \claw 31 2 3 44;
    \claw 33 2 3 22;
    \claw 33 2 3 24;
    \claw 33 2 3 42;
    \claw 33 2 3 44;
    \claw 33 4 3 22;
    \claw 33 4 3 24;
    \claw 33 4 3 42;
    \claw 33 4 3 44;
    \Description{
        Twenty-four sextuples that are compatible with the Z-shaped relation.
    }
    \caption{
        The number of sextuples $(a_1, a_2, b, a, b_1, b_2)$ satisfying
        $\ZZ(a_1, b)$, $\ZZ(a_2, b)$, $\ZZ(a, b)$, $\ZZ(a, b_1)$, and $\ZZ(a,
        b_2)$ is $\Mom{3}{\ZZ}{3} = 24$.
    }                                                         \label{fig:313}
\end{figure}

\section{The Old Dexterous Framework}                         \label{sec:old}

    This section reviews Abo Khamis, Nakos, Olteanu, and Suciu's dexterous
    $\ell_p$-bounds \cite{ANO24}.  They offer bounds like
    \eqref{triple}--\eqref{335} in a systematic way.  This general framework
    consists of three steps.

\subsection{A three-step framework}                        \label{sec:3steps}

    Step 0 (problem setup): There are several attributes, for instance, $A$
    is the set of sellers, $B$ is the set of merchandises, and $C$ is the set
    of buyers.  Between any pair of attributes there can be relations.  For
    instance, say $R(A, B)$ records if a seller sells a merchandise; it is a
    table whose rows are pairs $(a, b) \in A \times B$ where $a$ sells $b$.
    Between $B$ and $C$ there can a relation $S(B, C)$ that records if a
    buyer wants a merchandise.  And between $C$ and $A$ can be a relation
    $T(C, A)$ that records if a buyer lives near a seller.  Now, say we are
    interested in the join query $R(A, B) \land S(B, C) \land T(C, A)$.  The
    result of this query should be a table whose rows are triples $(a, b, c)$
    where $a$ sells $b$ that $c$ wants and they live near each other.  The
    challenge here is to estimate the number of rows before actually forming
    the table.

    Step 1: Imagine that the end table is there and $(X, Y, Z)$ is a row
    sampled uniformly at random.  Here, $X$ is the value in the $A$ column,
    $Y$ is the value in the $B$ column, and $Z$ is the value in the $C$
    column.  Let
    \[
        H(X, Y, Z) \coloneqq
        \sum_{x,y,z} P(x, y, z) \ln \frac{1}{P(x, y, z)}
    \]
    be the Shannon entropy of $(X, Y, Z)$, where $P$ denotes pmf.  Since $(X,
    Y, Z)$ is uniform, $H(X, Y, Z)$ coincides with $\ln \#_\Delta$, the
    logarithm of the join cardinality.

    Step 2: Now upper-bounding the join cardinality is equivalent to
    upper-bounding $H(X, Y, Z)$.  And upper-bounding $H(X, Y, Z)$ is a
    well-studied problem in information theory.  For instance, entropy is
    subadditive: $H(X, Y, Z) \le H(X) + H(Y) + H(Z)$.  Entropy also has chain
    rule and submodularity: $H(X, Y, Z) = H(X, Y) + H(Z | X, Y) \le H(X, Y) +
    H(Z|Y)$.  There are other off-the-shelf inequalities, such as Han's
    inequality $2 H(X, Y, Z) \le H(X, Y) + H(Y, Z) + H(Z, X)$ and Shearer's
    inequality.  Collectively these type of inequalities are called
    \emph{Shannon-type inequalities}.

    Step 3: As the previous step decomposes $H(X, Y, Z)$ into local terms
    like $H(X)$, $H(X, Y)$, and $H(Y|X)$, the last step is to bound local
    terms using invariants of the input relations.  Some examples include
    \eqref{p=0}--\eqref{p=p}, and \eqref{p1q} is our addition to this family.

    For the remainder of this section, we explain steps 2 and 3 in two
    separate subsections.  We then demonstrate how those lead to 
    \eqref{triple}--\eqref{335} in the last subsection.

\begin{figure}
    \centering
    \begin{tikzpicture}
        \begin{axis} [
            view={-30}{25}, clip=false, zmode=log,
            xlabel=$p$, ylabel=$q$, zlabel=$\Mom{p}{\ZZ}{q}$,
            zlabel style={rotate=-90, xshift=10pt}, zmax=100
        ]
            \def\f(#1,#2){2^(#1-1)+2^(#2-1)+2^(#1+#2-2)}
            \def\g(#1){ln(#1)}
            \def\preparetriangles{
                \def\trianglebuffer{}
                \foreach\x in{1,1.25,...,3.99}{
                    \foreach\y in{1,1.25,...,3.99}{
                        \computexyz&meta(\xa,\ya,\za)[\wa]=\f(\x,\y)
                        \computexyz&meta(\xb,\yb,\zb)[\wb]=\f(\x+0.25,\y)
                        \computexyz&meta(\xc,\yc,\zc)[\wc]=\f(\x,\y+0.25)
                        \pushtriangle{
                            (\xa,\ya,\za)[\wa]
                            (\xb,\yb,\zb)[\wb]
                            (\xc,\yc,\zc)[\wc]
                        }
                    }
                }
            }
            \addplottriangles
            \tikzset{
                1/.style={draw=yellow!40!black, dotted},
                2/.style={yellow!80!black, node contents=$\bullet$},
                3/.style={yellow!60!black}
            }
            \draw [1] (axis description cs: 0.2, 1)
                node (X) [3] {Figure~\ref{fig:313}}
                (X) -- (3, 3, 24) node [2];
            \draw [1] (axis description cs: 0, 0.1)
                node (X) [3] {Figure~\ref{fig:113}}
                (X) -- (1, 3, 9) node [2];
            \draw [1] (axis description cs: 0.9, 0)
                node (X) [3] {Figure~\ref{fig:311}}
                (X) -- (3, 1, 9) node [2];
            \draw [1] (axis description cs: 0.5, -0.1)
                node (X) [3] {Figure~\ref{fig:212}}
                (X) -- (2, 2, 8) node [2];
            \draw [1] (axis description cs: 0.2, -0.1)
                node (X) [3] {Figure~\ref{fig:111}}
                (X) -- (1, 1, 3) node [2];
        \end{axis}
    \end{tikzpicture}
    \Description{
        A plot of $\Mom{p}{\ZZ}{q}$.
    }
    \caption{
        A linear--linear--log plot of $\Mom{p}{\ZZ}{q}$.  Notice that the
        point corresponding to Figure~\ref{fig:212} is below the midpoints of
        the two corresponding to Figures \ref{fig:311} and \ref{fig:113}.
        This is explained by \eqref{holder}.
    }                                                      \label{fig:holder}
\end{figure}
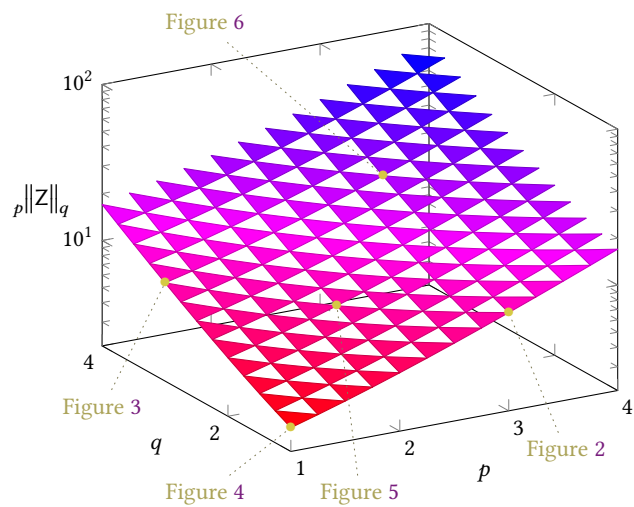

\subsection{Shannon-type inequalities}

    Shannon-type inequalities are a family of inequalities that are
    straightforward consequences of Jensen's inequality.  They are
    parameterized by subsets of variables and are good for generating upper
    bounds on an entropy, say $H(X, Y, Z)$, using simpler entropy terms such
    as $H(X)$, $H(X, Y)$, and $H(Y|X)$.  Note that there are also entropy
    inequalities not in the Shannon family.  Those are hard to enumerate and
    are of little use in a programming context like this.

    The one ``seed'' that generates all Shannon-type inequalities is
    submodularity.  Its proofs can be found in standard textbooks.
    
    \begin{lemma} [submodularity]                             \label{lem:sub}
        Let $X$, $Y$, and $Z$ be three random variables that may or may not
        correlate.  Then
        \[ H(X, Y) + H(Y, Z) \ge H(X, Y, Z) + H(Y). \]
        A common paraphrase is $H(Z|Y) \ge H(Z | X, Y)$, wherein $H(Z|Y)$
        denotes the \emph{conditional entropy} and is defined to be $H(Y, Z)
        - H(Y)$.  Another common paraphrase is $I(X; Z | Y) \ge 0$, wherein
        $I(X; Z | Y)$ denotes the \emph{conditional mutual information} and
        is defined to be $H(Z|Y) - H(Z | X, Y)$.
    \end{lemma}

    \begin{proof}
        Submodularity is a well-documented property of entropy; the following
        proof is included for completeness.  We group $(x, y, z)$ by $y$
        before applying Jensen's inequality to cancel $P(x, z | y)$:
        \begin{align*}
            \kern2em & \kern-2em
            H(X, Y, Z) + H(Y) - H(X, Y) - H(Y, Z) 
            \\ & = \sum_{x,y,z} P(x, y, z)
                \ln \frac{P(x, y) P(y, z)}{P(x, y, z) P(y)}
            \\ & = \sum_y P(y) \sum_{x,z} P(x, z | y)
                \ln \frac{P(x|y) P(z|y)}{P(x, z | y)}
            \\ & \le \sum_y P(y) \ln \sum_{x,z} P(x, z | y) \cdot
                \frac{P(x|y) P(z|y)}{P(x, z | y)}
            \\ & = \sum_y P(y) \ln \sum_{x,z} P(x|y) P(z|y)
            \\ & = \sum_y P(y) \ln(1) = 0
        \end{align*}
        This proves submodularity.
    \end{proof}

    Some special cases of submodularity have names.  For instance,
    monotonicity refers to $H(X, Y) \ge H(Y)$.  It is proved by letting $X =
    Z$ to get $2 H(X, Y) = H(X, Y) + H(X, Z) \ge H(X, Y, Z) + H(Y) = H(X, Y)
    + H(Y)$.  A paraphrase of this is $H(Y|X) \ge 0$.  Nonnegativity refers to
    $H(X) \ge 0$.  It is proved by letting $Y = 0$ in monotonicity.
    Subadditivity refers to $H(X, Z) \le H(X) + H(Z)$.  It is proved by
    letting $Y = 0$ to get $H(X) + H(Z) = H(X, Y) + H(Y, Z) \ge H(X, Y, Z) +
    H(Y) = H(X, Z)$.  A paraphrase of this is $H(Z|X) \le H(Z)$.

\subsection{Proof of dexterous bounds}

    In this subsection, we prove \eqref{p=p}.  As warm-ups, we prove special
    cases \eqref{p=0}--\eqref{p=oo} first.  To prove \eqref{p=0}, apply
    Jensen's inequality to cancel $P(a)$:
    \[
        H(X) = \sum_{a\in A} P(a) \ln \frac{1}{P(a)}
        \le \ln \sum_{a\in A} P(a) \cdot \frac{1}{P(a)} = \ln |A|
    \]
    To prove \eqref{p=1}, apply Jensen's inequality to cancel $P(a, b)$:
    \begin{align*}
        H(X, Y)
        & = \sum_{(a,b)\in R} P(a, b) \ln \frac{1}{P(a, b)}
        \\ & \le \ln  \sum_{(a,b)\in R} P(a, b) \cdot \frac{1}{P(a, b)}
        = \ln |R(A, B)|
    \end{align*}
    To prove \eqref{p=oo}, apply Jensen's inequality to cancel $P(b|a)$:
    \begin{align*}
        H(Y|X)
        & = \sum_{a\in A} P(a) \sum_{b:R(a,b)} P(b|a) \ln \frac{1}{P(b|a)}
        \\ & \le \sum_{a\in A} P(a) \ln \sum_{b:R(a,b)}
            P(b|a) \cdot \frac{1}{P(b|a)}
        \\ & = \sum_{a\in A} P(a) \ln \deg_R(a)
        \le \max_{a\in A} \ln \deg_R(a)
    \end{align*}
    We are now ready to prove \eqref{p=p}.

    \begin{proposition}                                       \label{pro:p11}
        (Lemma~4.1 of \cite{ANO24}) Let $p \ge 0$.  Let $R(A, B)$ be a
        relation with column types $A$ and $B$.  Let $(X, Y) \in R(A, B)$ be
        any random pair.  Then $H(X) + p H(Y|X) \le \lnMom{p}{R(A, B)}{1}$,
        i.e., \eqref{p=p} holds.
    \end{proposition}

    \begin{proof}
        The proof in \cite{ANO24} is reproduced here for self-containment.
        To begin, apply Jensen's inequality to cancel $P(b|a)$:
        \begin{align*}
            H(Y|X) 
            & = \sum_{a\in A} P(a) \sum_{b:R(a,b)}
                P(b|a) \ln \frac{1}{P(b|a)}
            \\ & \le \sum_{a\in A} P(a) \ln \sum_{b:R(a,b)}
                P(b|a) \cdot \frac{1}{P(b|a)}
            \\ & = \sum_{a\in A} P(a) \ln \deg_R(a)
        \end{align*}

        Hence,
        \begin{equation}
            p H(Y|X) \le \sum_{a\in A} P(a) \ln \deg_R(a)^p.     \label{pHYX}
        \end{equation}
        We also have, by definition,
        \begin{equation}
            H(X) = \sum_{a\in A} P(a) \ln \frac{1}{P(a)}.          \label{HX}
        \end{equation}
        Now combine \eqref{pHYX} and \eqref{HX} and apply Jensen's inequality
        to cancel $P(a)$:
        \begin{align*}
            H(X) + p H(Y|X)
            & \le \sum_{a\in A} P(a) \ln \frac{\ln \deg_R(a)^p}{P(a)}
            \\ & \le \ln \sum_{a\in A} P(a)
                \cdot \frac{\ln \deg_R(a)^p}{P(a)}
            \\ & = \ln \sum_{a\in A} \deg_R(a)^p = \lnMom{p}{R(A, B)}{1}.
        \end{align*}
        This concludes the proof of \eqref{p=p}.
    \end{proof}

\def\eqs{\texorpdfstring{{\eqref{triple}--\eqref{335}}}{old examples}}
\subsection{How to derive \eqs?}

    In this subsection, we demonstrate that \eqref{p=p} bounds yield
    \eqref{triple}--\eqref{335} systematically.  We start from
    \eqref{triple}.  Note that \eqref{triple} by itself is trivial, but
    proving it using the general framework makes a great example: By
    subadditivity and \eqref{p=0} we have
    \[
        \ln \#_\Delta = H(X, Y, Z)
        \le H(X) + H(Y) + H(Z)
        \le \ln |A| + \ln |B| + \ln |C|.
    \]
    Now \eqref{triple} follows after we remove logarithm.

    Next we prove \eqref{AGM}, which implies \eqref{2edge} and \eqref{half}.
    Without loss of generality, assume $u \ge v \ge w$.  Use monotonicity,
    subadditivity, and Han's inequality to derive
    \begin{align*}
        0 & \le (u - v) H(X, Y),
        \\ (v - w) H(X, Y, Z) & \le (v - w) H(X, Y) + (v - w) H(Y, Z),
        \\ 2w H(X, Y, Z) & \le w H(X, Y) + w H(Y, Z) + w H(Z, X).
    \end{align*} 
    Sum these inequalities to get
    \[ (v + w) H(X, Y, Z) \le u H(X, Y) + v H(Y, Z) + w H(Z, X). \]
    Note that $v + w \ge 1$ by assumption, so the left-hand side is at least
    $\ln \#_\Delta$.  By \eqref{p=1}, the right-hand side is at most $u \ln
    |R(A, B)| + v \ln |S(B, C)| + w \ln |T(C, A)|$.  This proves \eqref{AGM},
    and hence \eqref{2edge} and \eqref{half}.

    We proceed to proving \eqref{PANDA}.  Use chain rule and submodularity to
    conclude
    \begin{align}
        \ln \#_\Delta = H(X, Y, Z)
        & \le H(X, Y) + H(Z|X, Y)                                     \notag
        \\ & \le H(X, Y) + H(Z|Y).                             \label{chain}
    \end{align}
    We already saw $H(X, Y) \le \ln |R(A, B)|$ in \eqref{p=1} as well as
    $H(Z|Y) \le \ln \max_{b\in B} \deg_S(b)$ in \eqref{p=oo}, so the
    right-hand side of \eqref{chain} is at most $\ln |R(A, B)| + \ln
    \max_{b\in B} \deg_S(b)$.  This proves the fourth of the six inequalities
    in \eqref{PANDA}.  The other five follow by symmetry.

    Moving on to \eqref{222}, we repeat \eqref{chain} three times:
    \begin{align*}
        H(X, Y, Z) & \le H(X, Y) + H(Z|Y) = H(X) + H(Y|X) + H(Z|Y),
        \\ H(X, Y, Z) & \le H(Y, Z) + H(X|Z) = H(Y) + H(Z|Y) + H(X|Z),
        \\ H(X, Y, Z) & \le H(Z, X) + H(Y|X) = H(Z) + H(X|Z) + H(Y|X).
    \end{align*}
    Sum them to get
    \[
        3 H(X{,}Y{,}Z)
        \le H(X) + 2 H(Y|X) + H(Y) + 2 H(Z|Y) + H(Z) + 2H(X|Z).
    \]
    Apply \eqref{p=p} with $p = 2$, $2$, and $2$ to turn the right-hand side
    into $\lnMom{2}{R}{1} + \lnMom{2}{S}{1} + \lnMom{2}{T}{1}$, verifying
    \eqref{222}.

    Finally, we are left with \eqref{335}.  Sum
    \begin{align*}
        H(X, Y, Z) & \le H(X) + H(Y|X) + H(Z),
        \\ 2 H(X, Y, Z) & \le 2 H(Y|X) + 2 H(Z, X),
        \\ 3 H(X, Y, Z) & \le 3 H(Y|Z) + 3 H(Z, X)
    \end{align*}
    to get
    \[ 6 H(X, Y, Z) \le H(X) + 3 H(Y|X) + H(Z) + 3 H(Y|Z) + 5 H(Z, X). \]
    Apply \eqref{p=p} with $p = 3$, $3$, and $1$ to bound the right-hand side
    by
    $\ln\Mom{3}{R}{1} + \ln\Mom{1}{S}{3} + 5 \ln\Mom{1}{T}{1}$,
    confirming \eqref{335}.

\section{Our New Ambidextrous bounds}                         \label{sec:new}

    In this section, we assume the same three-step framework as was described
    in Section~\ref{sec:3steps}.  We improve the third step by strengthening
    \eqref{p=p} to \eqref{p1q} in the first subsection.  We then derive the
    example bound \eqref{345} in the subsection after that.  As a bonus, we
    explain why we can view bounding $\#_\Delta$ as a fractional covering
    problem on the Venn diagram.  Finally, we discuss the advantages of the
    new bounds in the third subsection and a challenge they face in the
    fourth subsection.

\subsection{Proof of the ambidextrous bounds}

    \begin{theorem} [main]                                    \label{thm:p1q}
        Let $p, q \ge 1$.  Let $R(A, B)$ be a relation with column types $A$
        and $B$.  Let $(X, Y) \in R(A, B)$ be any random pair.  Then $p
        H(Y|X) + I(X; Y) + q H(X|Y) \le \lnMom{p}{R(A, B)}{q}$, i.e.,
        \eqref{p1q} holds.
    \end{theorem}

    \begin{proof}
        First, apply Jensen's inequality to cancel $P(b|a)$:
        \begin{align*}
            H(Y|X)
            & = \sum_{(a,b)\in R} P(a, b) \ln \frac{1}{P(b|a)}
            \\ & = \sum_{a\in A} P(a) \sum_{b:R(a,b)}
                P(b|a) \ln \frac{1}{P(b|a)}
            \\ & \le \sum_{a\in A} P(a) \ln \sum_{b:R(a,b)}
                P(b|a) \cdot \frac{1}{P(b|a)}
            \\ & = \sum_{a\in A} P(a) \ln \deg_R(a).
        \end{align*}
        By symmetry, we have
        \begin{align}
            (p - 1) H(Y|X)
            & \le \sum_{(a,b)\in R} P(a, b) \ln \deg_R(a)^{p-1}, \label{p-1}
            \\ (q - 1) H(X|Y)
            & \le \sum_{(a,b)\in R} P(a, b) \ln \deg_R(b)^{q-1}. \label{q-1}
        \end{align}
        We also have, by definition,
        \begin{equation}
            H(Y|X) + I(X; Y) + H(X|Y)
            = \sum_{(a,b)\in R} P(a, b) \ln \frac{1}{P(a, b)}.    \label{HIH}
        \end{equation}
        Now combine \eqref{p-1}--\eqref{HIH} and apply Jensen's inequality to
        cancel $P(a, b)$:
        \begin{align*}
            \kern2em & \kern-2em
            p H(Y|X) + I(X; Y) + q H(X|Y)
            \\ & \le \sum_{(a,b)\in R} P(a, b) \ln
                \frac{\deg_R(a)^{p-1} \deg_R(b)^{q-1}} {P(a, b)}
            \\ & \le \ln \sum_{(a,b)\in R} P(a, b)
                \cdot \frac{\deg_R(a)^{p-1} \deg_R(b)^{q-1}} {P(a, b)}
            \\ & = \ln \sum_{(a,b)\in R} \deg_R(a)^{p-1} \deg_R(b)^{q-1}
            = \lnMom{p}{R(A, B)}{q}.
        \end{align*}
        This concludes the proof of \eqref{p1q}.
    \end{proof}

\def\eqs{\texorpdfstring{\eqref{345}}{new example}}
\subsection{How to derive \eqs\ and more bounds?}

    Our new ambidextrous bounds \eqref{p1q} lead to many new bounds like
    \eqref{345}.  In this subsection, we first demonstrate how to derive
    \eqref{345} and discuss how to generate more.

    The derivation starts with the six different ways to permute
    \eqref{chain}:
    \begin{align*}
        H(X, Y, Z) & \le H(X, Y) + H(Z|Y)
        \\ H(X, Y, Z) & \le H(Y, Z) + H(X|Z)
        \\ H(X, Y, Z) & \le H(Z, X) + H(Y|X)
        \\ 2 H(X, Y, Z) & \le 2 H(Y, X) + 2 H(Z|X)
        \\ 2 H(X, Y, Z) & \le 2 H(Z, Y) + 2 H(X|Y)
        \\ 2 H(X, Y, Z) & \le 2 H(X, Z) + 2 H(Y|Z)
    \end{align*}
    The right-hand sides of the inequalities above have the same sum as the
    left-hand sides of the inequalities below:
    \begin{align*}
        3 H(X, Y) + 2 H(X|Y) + H(Y|X) & = 4 H(Y|X) + 3 I(X{;}Y) + 5 H(X|Y)
        \\ 3 H(Y, Z) + 2 H(Y|Z) + H(Z|Y) & = 4 H(Z|Y) + 3 I(Y{;}Z) + 5 H(Y|Z)
        \\ 3 H(Z, X) + 2 H(Z|X) + H(X|Z) & = 4 H(X|Z) + 3 I(Z{;}X) + 5 H(Z|X)
    \end{align*}
    The right-hand sides of the inequalities above can be bounded by
    \eqref{p1q}, resulting in
    \[
        9 H(X, Y, Z)
        \le 3 \lnMom{4/3}{R}{5/3}
        + 3 \lnMom{4/3}{S}{5/3}
        + 3 \lnMom{4/3}{T}{5/3}.
    \]
    This concludes the proof of \eqref{345}.

    We now give an easy criterion that determines if a bound on $H(X, Y, Z)$
    is valid.  This can be used to generate more bounds like \eqref{345} with
    ease.

    \begin{lemma} [Venn criterion]                           \label{lem:venn}
        Consider a linear combination of these nine terms
        \begin{gather*}
            H(Y|X)\;,\; I(X; Y)\;,\; H(X|Y)\;,\; \\
            H(Z|Y)\;,\; I(Y; Z)\;,\; H(Y|Z)\;,\; \\
            H(X|Z)\;,\; I(Z; X)\;,\; H(Z|X)\;,\;
        \end{gather*}
        where \textbf{the coefficient of an ``$I$'' is no greater than the
        coefficients of the two ``$H$'' on the same line.} This linear
        combination upper-bounds $H(X, Y, Z)$ if and only if it forms a
        fractional covering of the Venn diagram in the sense of
        Figure~\ref{fig:venn}.  That is, it upper-bounds $H(X, Y, Z)$ iff
        changing basis to
        \begin{gather*}
            H(X | YZ)\;,\; H(Y | ZX)\;,\; H(Z | XY)\;,\; \\
            I(X; Y | Z)\;,\; I(Y; Z | X)\;,\; I(Z; X | Y)\;,\; \\
            I(X; Y; Z)\;,\;
        \end{gather*}
        results in coefficients greater than or equal to $1$.
    \end{lemma}

    \begin{proof}
        [Proof of necessity (i.e., the only if part, $\Rightarrow$)]
        Denote by $\alpha, \beta, \gamma, \delta,
        \varepsilon, \zeta, \eta \in \{0, 1\}^7$ seven iid Bernoulli random
        variables.  Let $X = (\alpha, \delta, \zeta, \eta)$; let $Y = (\beta,
        \delta, \varepsilon, \eta)$; and let $Z = (\gamma, \varepsilon,
        \zeta, \eta)$.  Clearly $\alpha, \dotsc, \eta$ correspond to $H(X |
        YZ)$, $\dotsc$, $I(X; Y; Z)$, in that order.  If any coefficient is
        $< 1$, say that of $H(X | YZ)$, we will let $H(\alpha) = 1/2$ and
        $H(\beta) = H(\gamma) = H(\delta) = H(\varepsilon) = H(\zeta) =
        H(\eta) = 0$.  This makes $H(X, Y, Z) = H(\alpha) = 1/2$ while making
        the linear combination $< 1/2$.  The linear combination is thus not
        an upper bound.  The same trick applies to other six coefficients
        being $< 1$.  This finishes the proof of necessity.
    \end{proof}

    \begin{proof}
        [Proof of sufficiency (i.e., the if part, $\Leftarrow$)]
        We recycle Greek letters $\alpha, \dotsc, \eta$ and let them denote
        the coefficients in the Venn basis $H(X | YZ), \dotsc, I(X; Y; Z)$,
        in that order.  The proof is divided into two cases.  The first case
        is that $\eta$, the coefficient of $I(X; Y; Z)$, is exactly $1$.  In
        this case, since other $I$'s and $H$'s are nonnegative (cf.\
        lemma~\ref{lem:sub} and the paragraph thereafter), a linear
        combination with $\alpha, \dotsc, \zeta \ge 1$ is at least the linear
        combination with $\alpha, \dotsc, \zeta = 1$.  The latter coincides
        with $H(X, Y, Z)$ and finishes the case.

        The second case is $\eta > 1$.  This is where we use the
        \textbf{boldfaced} condition: For any $p, q \ge 1$, we can see that
        $p H(Y|X) + I(X; Y) + q H(X|Y)$ contribute more to $\delta,
        \varepsilon$ and $\zeta$ than to $\eta$.  The same applies to $p
        H(Z|Y) + I(Y; Z) + q H(Y|Z)$ and $p H(X|Z) + I(Z; X) + q H(Z|X)$.
        Hence, $\eta \le \delta, \varepsilon, \zeta$.  This means that we can
        decrease the linear combination by lowering the supply of $I(X; Y)$,
        $I(Y; Z)$, and $I(Z; X)$.  We can do so until $\eta = 1$, which
        reduces the problem to the first case.  This concludes the proof of
        sufficiency and hence the proof of Lemma~\ref{lem:venn}.
    \end{proof}

    Note that ambidextrous bounds always satisfy the \textbf{boldfaced}
    condition because Theorem~\ref{thm:p1q} requires $p, q \ge 1$.
    Therefore, Lemma~\ref{lem:venn} implies that a combination of
    ambidextrous bounds is a valid upper bound on $H(X, Y, Z)$ iff the
    coefficients are $\ge 1$ after changing to the Venn basis---that is, iff
    it forms a fractional covering of the Venn diagram.  As a corollary,
    Lemma~\ref{lem:venn} provides a universal proof to the Shannon-inequality
    part of \eqref{triple}--\eqref{335} and \eqref{345}.

\subsection{Advantage over the old bounds}

    As mentioned in the introduction, the bivariate moments of bi-degree
    sequences are a generalization of the $p$-norms of degree sequences.  As
    the former forms a strictly larger family of pessimistic estimates,
    chances are that some are strictly tighter than all old estimates.  That
    being said, this type of set-theoretically tautological improvements are
    frequently complexity trade-offs in disguise---should one be unable to
    afford to enumerate the larger family, the existence of tighter estimates
    therein helps very little.

    So, in the introduction, we argued against this by pointing out a pathway
    to enjoy new bounds for free.  The argument goes: If $R$ appears twice,
    they must come from two applications of \eqref{p=p} to $w H(X) + w p
    H(Y|X)$ and $(1 - w) H(Y) + (1 - w) q H(X|Y)$, which yield
    $\Mom{p}{R}{1}^w \cdot \Mom{1}{R}{q}^{1-w}$.  When that happens, we can
    apply \eqref{p1q} instead to yield $\Mom{w p+(1-w)}{R}{w+(1-w)q}$.  The
    latter ties or tightens the former due to \eqref{holder}.

    But how likely is it that $R$ appears twice?  The answer is that we can
    sometimes \emph{force} it.  Take $\#_\Delta$ as an example.  Suppose that
    $R = S = T$ are the same relation and that it is symmetric.  A generic
    variant of \eqref{triple}--\eqref{335} that look like
    \[
        \#_\Delta \le
        \Mom{p}{R}{1}^u + \Mom{q}{S}{1}^v + \Mom{r}{T}{1}^w
    \]
    can be mirrored to
    \[
        \#_\Delta \le
        \Mom{1}{R}{p}^u + \Mom{1}{S}{q}^v + \Mom{1}{T}{r}^w.
    \]
    Now that we see both $\Mom{p}{R}{1}$ and $\Mom{1}{R}{p}$, Hölder's
    inequality \eqref{holder} applies.  The same can be said to counting
    quadrilaterals, pentagons, hexagons, etc.

    We now prove \eqref{holder} by proving the general case \eqref{holder2}.

    \begin{theorem} [Hölder and convexity]                 \label{thm:holder}
        Let $p, q, r, s\ge 1$ and $0 < w < 1$.  Let $R(A, B)$ be a relation.
        Then $\Mom{p}{R}{q}^w \cdot \Mom{r}{R}{s}^{1-w} \ge \Mom{w
        p+(1-w)r}{R}{w q+(1-w)s}$, i.e., \eqref{holder} and \eqref{holder2}
        hold.  As a consequence, $\lnMom{p}{R}{q}$ is convex in $(p, q)$.
    \end{theorem}

    \begin{proof}
        \def\deg_R{d}
        To save space, $d$ means $\operatorname{deg}_R$ in this proof.  Apply
        Hölder's inequality to the left-hand side of the target inequality:
        \begin{align*}
            \Mom{p}{R}{q}^w \cdot \Mom{r}{R}{s}^{1-w}
            & = \Bigl( \sum_{(a,b)\in R}
                \deg_R(a)^{p-1} \deg_R(b)^{q-1} \Bigr)^w
                \Bigl( \sum_{(a,b)\in R}
                \deg_R(a)^{r-1} \deg_R(b)^{s-1} \Bigr)^{1-w}
            \\ & \ge \sum_{(a,b)\in R}
                \bigl( \deg_R(a)^{p-1} \deg_R(b)^{q-1} \bigr)^w
                \bigl( \deg_R(a)^{r-1} \deg_R(b)^{s-1} \bigr)^{1-w}
            \\ & = \sum_{(a,b)\in R}
                \deg_R(a)^{w p+(1-w)r-1} \deg_R(b)^{wq+(1-w)s-1}
            \\ & = \Mom{wp+(1-w)r}{R}{wq+(1-w)s}
        \end{align*}
        This proves \eqref{holder} and \eqref{holder2}.  Now take $\log$ to
        arrive at $w \lnMom{p}{R}{q} + (1 - w) \lnMom{r}{R}{s} \ge
        \lnMom{wp+(1-w)r}{R}{wq+(1-w)s}$, which is one way to define/verify
        convexity.
    \end{proof}

    The convexity implies that the best dexterous bound cannot have both
    $\Mom{p_1}{R}{1}$ and $\Mom{p_2}{R}{1}$ for some $p_1 < p_2$, for
    otherwise we can always replace them with $\Mom{p}{R}{1}$ for some $p
    \in [p_1, p_2]$ and be better off.  However, $R$ can still appear twice
    in the form of $\Mom{p}{R}{1}$ and $\Mom{1}{R}{q}$ for some $p, q \ge
    1$, which can only be improved by inventing ambidextrous bounds.

    Now that ambidextrous bounds are invented, the convexity then implies
    that the best bound is always a single moment term.  Not only is it
    tighter than all linear combinations of $\lnMom{p_i}{R}{1}$ and
    $\lnMom{1}{R}{q_i}$, it is also tighter than all linear combinations of
    the general building block $\lnMom{p_i}{R}{q_i}$.  The preceding
    argument is made precise in the following corollary.

    \begin{corollary} [simple is optimal]                  \label{cor:single}
        Consider the task of bounding $p H(Y|X) + I(X; Y) + q H(X|Y)$, where
        $p, q \ge 1$.        Among all possible ways to apply
        \eqref{p1q}---including the special cases
        \eqref{p=0}--\eqref{p=p}---to obtain a linear combination of
        $\lnMom{p_i}{R}{q_i}$, the tightest one must be a single term $w
        \lnMom{p/w}{R}{q/w}$ for some $w \in [1, \min(p, q)]$.
    \end{corollary}

    \begin{proof}
        The very first step is to verify that $w \lnMom{p/w}{R}{q/w}$ is a
        valid upper bound for any $w \in [1, \min(p, q)]$.  This is simply
        because $p H(Y|X) + I(X; Y) + q H(X|Y) \le p H(Y|X) + w I(X; Y) + q
        H(X|Y) \le w \lnMom{p/w}{R}{q/w}$.  That is, we can always add more
        $I(X; Y)$ before applying \eqref{p1q}.

        Next, let us consider a bound
        \[
            p H(Y|X) + I(X; Y) + q H(X|Y)
            \le \sum_i w_i \lnMom{p_i}{R}{q_i}
        \]
        obtained  by applying \eqref{p1q} like this:
        \begin{gather*}
            w_1 p_1 H(Y|X) + w_1 I(X; Y) + w_1 q_1 H(X|Y)
            \le w_1 \lnMom{p_1}{R}{q_1} \\
            w_2 p_2 H(Y|X) + w_2 I(X; Y) + w_2 q_2 H(X|Y)
            \le w_2 \lnMom{p_2}{R}{q_2} \\
            w_3 p_3 H(Y|X) + w_3 I(X; Y) + w_3 q_3 H(X|Y)
            \le w_3 \lnMom{p_3}{R}{q_3} \\
            \vdots \kern5em \vdots \kern5em \vdots \kern6em \vdots
        \end{gather*}
        Fix $w \coloneqq \sum_i w_i$ to be the total weight.  We next want to
        show that the multi-term bound $\sum_i w_i \lnMom{p_i}{R}{q_i}$ is
        $\ge$ the single-term bound $w \lnMom{p/w}{R}{q/w}$.

        To show that, apply Theorem~\ref{thm:holder} to obtain $\sum_i w_i
        \lnMom{p_i}{R}{q_i} \ge w \lnMom{r/w}{R}{s/w}$, where $r \coloneqq
        \sum_i w_i p_i$ and $s \coloneqq \sum_i w_i q_i$ are weighted sums.
        It remains to show $\Mom{r/w}{R}{s/w} \ge \Mom{p/w}{R}{q/w}$.

        Note that $\Mom{p}{R}{q}$ is monotonic in $p$ and in $q$.  So the
        proof boils down to showing $r \ge p$ and $s \ge q$.  This is already
        the case because we used $\sum_i w_i p_i H(Y|X)$ to upper bound $p
        H(Y|X)$ and $\sum_i w_i q_i H(X|Y)$ to upper bound $q H(X|Y)$.  This
        concludes the proof of optimality.
    \end{proof}

    Corollary~\ref{cor:single} implies that, after we find a good bound using
    an approximate minimization algorithm, it (should) only contains one
    moment term per input relation.  Therefore, if we ever want to
    double-check this bound, generate a certificate, or refine it even
    further, it would be very easy to do so.

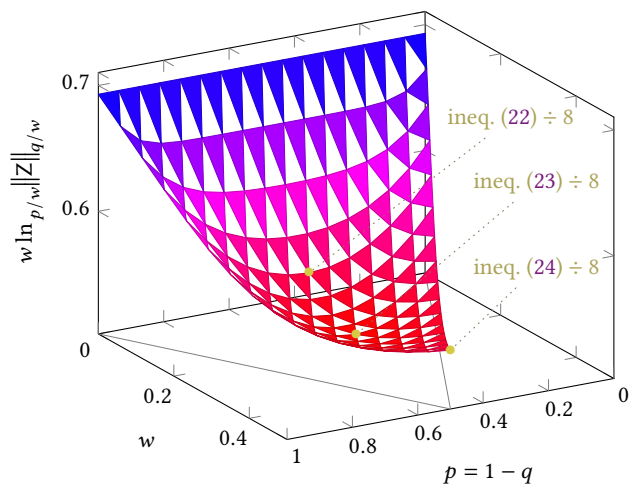
\begin{figure}
    \centering
    \begin{tikzpicture}
        \begin{axis} [
            view={-30}{25}, clip=false,
            xlabel={$p = 1 - q$}, ylabel={$w$},
            x dir=reverse, y dir=reverse,
            zlabel={$w\lnMom{p/w}{\ZZ}{q/w}$},
        ]
            \def\f(#1,#2){#2*ln(2^(#1/#2-1)+2^((1-#1)/#2-1)+2^(1/#2-2))}
            \def\fb(#1,#2){ln(2)}
            \def\g(#1){#1}
            \def\preparetriangles{
                \pgfkeys{/pgf/fpu,/pgf/fpu/output format=sci}
                \def\trianglebuffer{}
                \foreach\x in {1,...,16}{
                    \foreach\y in{1,...,\x}{
                        \pgfmathsetmacro\p{(\x+\y-1)/32}
                        \pgfmathsetmacro\w{(\x-\y)/32}
                        \computexyz&meta(\xa,\ya,\za)[\wa]=\f(\p,\w+1/32)
                        \computexyz&meta(\xb,\yb,\zb)[\wb]=\f(\p-1/32,\w)
                        \computexyz&meta(\xc,\yc,\zc)[\wc]=\f(\p+1/32,\w)
                        \pgfmathsetmacro\zb{\w==0?ln(2):\zb}
                        \pgfmathsetmacro\zc{\w==0?ln(2):\zc}
                        \pgfmathsetmacro\wb{\w==0?ln(2):\wb}
                        \pgfmathsetmacro\wc{\w==0?ln(2):\wc}
                        \pushtriangle{
                            (\xa,\ya,\za)[\wa]
                            (\xb,\yb,\zb)[\wb]
                            (\xc,\yc,\zc)[\wc]
                        }
                    }
                }
                \pgfkeys{/pgf/fpu=false}
            }
            \draw [line width=0.2, gray] (rel axis cs: 0, 1, 0) --
                (rel axis cs: 0.5, 0, 0) -- (rel axis cs: 1, 1, 0);
            \addplottriangles
            \tikzset{
                1/.style={draw=yellow!40!black, dotted},
                2/.style={yellow!80!black, node contents=$\bullet$},
                3/.style={yellow!60!black, align=center}
            }
            \draw [1] (axis description cs: 0.8, 0.75)
                node (X) [3] {ineq.\ \eqref{414}${}\div 8$}
                (X) -- (0.5, 1/8, 0.5477533293) node [2];
            \draw [1] (axis description cs: 0.85, 0.6)
                node (X) [3] {ineq.\ \eqref{424}${}\div 8$}
                (X) -- (0.5, 2/8, 0.5198603854) node [2];
            \draw [1] (axis description cs: 0.85, 0.4)
                node (X) [3] {ineq.\ \eqref{444}${}\div 8$}
                (X) -- (0.5, 4/8, 0.5493061443) node [2];
        \end{axis}
    \end{tikzpicture}
    \Description{
        A plot of $w \lnMom{p/w}{\ZZ}{q/w}$
    }
    \caption{
        A plot of $w \lnMom{p/w}{\ZZ}{q/w}$ after normalizing $p + q$ to
        $1$, i.e., the ``per conditional entropy'' bound.  The range of $w$
        is $[0, \min(p, q)]$.  Observe that the plot is not monotonic in $w$,
        so we might want to add more $I(X; Y)$ before \eqref{p1q}.
    }                                                      \label{fig:convex}
\end{figure}

\subsection{An unwanted degree of freedom}

    Corollary~\ref{cor:single} does not specify which $w$ minimizes $w
    \lnMom{p/w}{R}{q/w}$.  It is usually $w = 1$ but, interestingly, not
    always.  When $p$ and $q$ are too large (say $\ge 3$), ambidextrous
    bounds tend to overshoot too much to the point that adding $I(X; Y)$ to
    bring down $p$ and $q$ is beneficial.  A concrete example is below.

    \begin{example}
        Recall the Z-shaped relation $\ZZ = \{ (1, 2), (3, 2), (3, 4) \}$.
        Consider $(X, Y)$ being a random row of $\ZZ$.  Then \eqref{p1q} with
        $(p, q) = (3, 3)$ yields
        \begin{equation}
            4 H(X|Y) + I(X; Y) + 4 H(Y|X)
            \le \lnMom{4}{\ZZ}{4} = \ln 80.                      \label{414}
        \end{equation}
        At the same time, \eqref{p1q} with $(p, q) = (2, 2)$ yields
        (cf.\ Figure~\ref{fig:212})
        \begin{equation}
            4 H(X|Y) + 2 I(X; Y) + 4H(Y|X)
            \le 2 \lnMom{2}{\ZZ}{2}= \ln 64.                     \label{424}
        \end{equation}
        Moreover, \eqref{p1q} with $(p, q) = (1, 1)$ yields
        \begin{equation}
            4 H(X|Y) + 4 I(X; Y) + 4H(Y|X)
            \le 4 \lnMom{1}{\ZZ}{1}= \ln 81.                     \label{444}
        \end{equation}
        To summarize, the relaxation with twice as much $I(X; Y)$ benefits
        the bound.  But too much $I(X; Y)$ is still harmful.  See
        Figure~\ref{fig:convex} for the relative positions of
        \eqref{414}--\eqref{444}.
    \end{example}

    Note that the non-monotonicity in $w$ is a new feature of our
    ambidextrous bounds.  The old dexterous bounds do not need to be open to 
    the possibility of adding $I(X; Y)$ for improvement because $w \in [1,
    \min(p, q)] = [1, 1]$.  Very interesting implications follow:
    \begin{itemize}
        \item Even if we always stick to $w = 1$, ambidextrous bounds are
            still equivalent to or better than dexterous bounds.  (This is
            due to \eqref{holder}.)
        \item Anytime the minimizer $w$ is not $1$ but strictly greater, the
            ambidextrous bound at that $w$ is strictly better than the
            ambidextrous bound at $w = 1$.
        \item Therefore, any nontrivial minimizer $w$ gives rise to a
            strictly better bound over dexterous ones.
    \end{itemize}
    That is to say, as much as finding the optimal $w$ is challenging, it
    being nontrivial to find is a source of improvement over the old
    dexterous framework.

    To conclude this subsection, we offer a relief that $w
    \lnMom{p/w}{R}{q/w}$ is convex in $w$, so finding the best $w$ is less
    demanding than it sounds.

    \begin{proposition}                                    \label{pro:convex}
        $w \lnMom{p/w}{R}{q/w}$ is convex in $w \in [1, \min(p, q)]$.
    \end{proposition}

    \begin{proof}
        For a fixed relation $R$, $F(p, q, w) \coloneqq w
        \lnMom{p/w}{R}{q/w}$ is a trivariate smooth function.  $F$ is convex
        in $(p, q)$ by Theorem~\ref{thm:holder}.  $F$ is also homogeneous of
        degree $1$ in the sense that, for any $\lambda > 0$, $F(\lambda p,
        \lambda q, \lambda w) = \lambda F(p, q, w)$.  These characterize the
        Hessian matrix $\nabla^2 F$ at $(p, q, w)$: It has an eigenvalue $0$
        paired with eigenvector $(p, q, w)$; and the other two eigenvalues
        are positive.  As a result, $\nabla^2 F$ is positive semi-definite.
        Now vary $w$.  As we are heading in $(0, 0, 1)$, not in the same
        direction as the zero eigenvalue, we will always experience a
        strictly positive second derivative.  Thus, $F$ is (strictly) convex
        when $w$ alone varies.
    \end{proof}

    See Figure~\ref{fig:convex} for a slice of $w \lnMom{p/w}{R}{q/w}$ at $p
    + q = 1$, i.e., the bound when we normalize the total amount of $H(Y|X)$
    and $H(X|Y)$ to $1$, hence the name ``per conditional entropy bound''.
    The plot shows that increasing $w$ can sometimes lead to tighter bounds.
    But since the function is convex in $w$, we can stop searching after we
    see the slope $(\partial/\partial w) w \lnMom{p/w}{R}{q/w}$ turns
    positive.
\begin{tikzpicture} [overlay]
    \path foreach \i in {0, ..., 2} {(\i*120 + 90: 2pt) coordinate (c\i)};
    \path foreach \j in {0, ..., 3} {(22.5 + \j*90: 2pt) coordinate (d\j)};
    \path foreach \k in {0, ..., 4} {(\k*72: 2pt) coordinate (\k)};
\end{tikzpicture}
\def\setpic#1#2{\pgfdeclareplotmark{#2}{\pgftext{\copy#1}}}
\def\deficon#1#2#3#4{%{#vertices}{#edges}{icon name}{drawing code}
    \newbox\iconbox\setbox\iconbox=
    \hbox{\tikz[every path/.style={line width=0.2pt, red!#20!blue}]{#4}}
    \expandafter\setpic\expandafter{\the\iconbox}{#3}
}
\deficon{3}{2}{path3}{\draw(c2)--(c0)--(c1);}
\deficon{3}{3}{K3}{\draw(c2)--(c0)--(c1)--cycle;}
\deficon{4}{3}{claw}{\draw(d0)--(d1) (d0)--(d2) (d0)--(d3);}
\deficon{4}{3}{path4}{\draw(d0)--(d1)--(d2)--(d3);}
\deficon{4}{4}{pan3}{\draw(d3)--(d1)--(d2)--(d3)--(d0);}
\deficon{4}{4}{cycle4}{\draw(d0)--(d1)--(d2)--(d3)--(d0);}
\deficon{4}{5}{fan2}{\draw(d0)--(d1)--(d2)--(d3)--(d0) (d0)--(d2);}
\deficon{4}{6}{K4}{\draw(d0)--(d1)--(d2)--(d3)--(d0)--(d2) (d1)--(d3);}
\deficon{5}{4}{K14}{\draw(0)--(1) (0)--(2) (0)--(3) (0)--(4);}
\deficon{5}{4}{chair}{\draw(1)--(2)--(3)--(4) (0)--(2);}
\deficon{5}{4}{path5}{\draw(3)--(4)--(0)--(1)--(2);}
\deficon{5}{5}{cricket}{\draw(2)--(3)--(0)--(2) (1)--(0)--(4);}
\deficon{5}{5}{pan4}{\draw(3)--(4)--(0)--(1)--(2) (1)--(3);}
\deficon{5}{5}{bull}{\draw(3)--(4)--(0)--(1)--(2) (1)--(4);}
\deficon{5}{5}{pan4c}{\draw(3)--(4)--(0)--(1)--(2) (0)--(3);}
\deficon{5}{5}{cycle5}{\draw(0)--(1)--(2)--(3)--(4)--(0);}
\deficon{5}{6}{dart}{\draw(3)--(4)--(0)--(1)--(2) (1)--(4)--(2);}
\deficon{5}{6}{K23}{\draw(3)--(4)--(0)--(1)--(2) (1)--(3) (2)--(4);}
\deficon{5}{6}{butterfly}{\draw(3)--(4)--(0)--(1)--(2) (2)--(0)--(3);}
\deficon{5}{6}{house}{\draw(0)--(1)--(2)--(3)--(4)--(0) (1)--(4);}
\deficon{5}{6}{kite}{\draw(3)--(4)--(0)--(1)--(2) (0)--(2)--(4);}
\deficon{5}{7}{K3u2K1c}{\draw(3)--(4)--(0)--(1)--(2)--(4)--(3)--(1);}
\deficon{5}{6}{fan3}{\draw(0)--(1)--(2)--(3)--(4)--(0) (2)--(0)--(3);}
\deficon{5}{7}{clawuK1c}{\draw(3)--(4)--(0)--(1)--(2) (0)--(3)--(1)--(4);}
\deficon{5}{7}{P2uP3c}{\draw(0)--(1)--(2)--(3)--(4)--(0) (1)--(3) (2)--(4);}
\deficon{5}{8}{P3u2K1c}{\draw(0)--(1)--(2)--(3)--(4)--(0)(2)--(4)--(1)--(3);}
\deficon{5}{8}{wheel4}{\draw(2)--(0)--(1)--(2)--(3)--(4)--(0)--(2)(1)--(4);}
\deficon{5}{9}{K5_e}{\draw(1)--(2)--(3)--(4)--(0)--(1)--(3)--(0)--(2)--(4);}
\deficon{5}{10}{K5}{\draw(1)--(2)--(3)--(4)--(0)--(1)--(3)--(0)--(2)--(4)--(1);}

\section{Simulations}                                        \label{sec:simu}

    In this section, we present simulation results.

    For real-world relations, we use undirected graphs from the Stanford
    Large Network Dataset Collection (SNAP), with a few exceptions:
    \begin{itemize}
        \item Datasets \texttt{as-733}, \texttt{Oregon-1}, and
            \texttt{Oregon-2} are parameterized by dates, so only the latest
            ones are used.
        \item Large files \texttt{com-Friendster}, \texttt{com-Orkut},
            \texttt{com-LiveJournal}, \texttt{as-Skitter}, and
            \texttt{twitch-gamers} are excluded.
        \item JSON files \verb|deezer_edges|, \verb|twitch_edges|,
            \verb|git_edges|, and \verb|reddit_edges| are skipped for
            crashing our text editor too many times.
    \end{itemize}

    For queries, we consider all simple connected graphs with three, four and
    five vertices.  Two have three vertices:
    \tikz[scale=2]{\pgfuseplotmark{path3}}
    \tikz[scale=2]{\pgfuseplotmark{K3}}.
    Six have four vertices:
    \tikz[scale=2]{\pgfuseplotmark{claw}}
    \tikz[scale=2]{\pgfuseplotmark{path4}}
    \tikz[scale=2]{\pgfuseplotmark{pan3}}
    \tikz[scale=2]{\pgfuseplotmark{cycle4}}
    \tikz[scale=2]{\pgfuseplotmark{fan2}}
    \tikz[scale=2]{\pgfuseplotmark{K4}}.
    Twenty-one have five vertices:
    % 4 edges
    \tikz[scale=2]{\pgfuseplotmark{K14}}
    \tikz[scale=2]{\pgfuseplotmark{chair}}
    \tikz[scale=2]{\pgfuseplotmark{path5}}
    % 5 edges
    \tikz[scale=2]{\pgfuseplotmark{cricket}}
    \tikz[scale=2]{\pgfuseplotmark{pan4}}
    \tikz[scale=2]{\pgfuseplotmark{bull}}
    \tikz[scale=2]{\pgfuseplotmark{pan4c}}
    \tikz[scale=2]{\pgfuseplotmark{cycle5}}
    % 6 edges
    \tikz[scale=2]{\pgfuseplotmark{dart}}
    \tikz[scale=2]{\pgfuseplotmark{K23}}
    \tikz[scale=2]{\pgfuseplotmark{butterfly}}
    \tikz[scale=2]{\pgfuseplotmark{house}}
    \tikz[scale=2]{\pgfuseplotmark{kite}}
    % 7 edges
    \tikz[scale=2]{\pgfuseplotmark{K3u2K1c}}
    \tikz[scale=2]{\pgfuseplotmark{fan3}}
    \tikz[scale=2]{\pgfuseplotmark{clawuK1c}}
    \tikz[scale=2]{\pgfuseplotmark{P2uP3c}}
    % 8 edges
    \tikz[scale=2]{\pgfuseplotmark{P3u2K1c}}
    \tikz[scale=2]{\pgfuseplotmark{wheel4}}
    % 9 edges
    \tikz[scale=2]{\pgfuseplotmark{K5_e}}
    % 10 edges
    \tikz[scale=2]{\pgfuseplotmark{K5}}.
    Their colors depend monotonically on the numbers of edges.  See also
    \cite[A001349]{OEIS} and references therein.

    Now the task boils down to counting graph homomorphisms from one query
    graph to one SNAP dataset, as it is the most natural object degree
    sequences can bound.  Therefore, for the number of quadrilaterals
    \tikz[scale=2]{\pgfuseplotmark{cycle4}}, $(a, b, a, b)$ counts if $R(a,
    b)$ and $(a, a, a, a)$ counts if $R(a, a)$.  This is different from
    \cite{SuL20}'s counting of subgraphs, which are injective homomorphisms.
 
    Results are presented in Figures~\ref{fig:plots1}--\ref{fig:plots3}, with
    the name of each dataset in the lower right corner of each plot.  The
    horizontal axis is the actual count of homomorphisms; the vertical axis
    is for upper bounds.  In each plot, each of the icons
    \tikz[scale=2]{\pgfuseplotmark{path3}},
    \tikz[scale=2]{\pgfuseplotmark{K3}}, \ldots,
    \tikz[scale=2]{\pgfuseplotmark{K5}} appears twice.
    The one above is dexterous and the one below ambidextrous (recall that
    the latter is provably tighter).  The area below the diagonal (i.e., $y <
    x$) is shaded in yellow; it helps visualize how far off the bounds are.
    See also Figure~\ref{fig:average} for the geometric mean over all
    datasets.
    
    A clear trend is that dexterous bounds are almost a function of the
    numbers of vertices (the corresponding icons line up horizontally).  But
    query graphs with more edges have fewer homomorphisms into a dataset so
    dexterous bounds are less tight for denser graphs.  On the other hand,
    ambidextrous bounds are more sensitive to the numbers of edges---the
    corresponding icons still line up, but not as much---and yield tighter
    bounds.

    Another way to compare is to look at $\log_{10}$ of relative errors in
    Figures \ref{fig:plots4}--\ref{fig:plots6}.  After averaging over all
    datasets, we plot them in Figure~\ref{fig:relative}.  There is a linear
    trend saying that ambidextrous overshoot by $x^{0.7481}$-fold when
    dexterous overshoot by $x$-fold.  Here, $0.7481$ is the slope of the
    origin-passing least-square line.  It has an impressive coefficient of
    determination ($R^2$) of $0.9870$.

    As for implementation, we commented earlier that Section~\ref{sec:3steps}
    is a convex optimization problem.  But prior works (\textsc{SafeBound}
    \cite{DSB22,DSB23,DSB25} and \textsc{LpBound} \cite{ZMA25,MZA25}), use no
    convexity machinery.  Instead, they precomputed the norms at discrete
    values $p \in \{1, \ldots, 10, \infty\}$ and run a linear program
    constrained by submodularity.  We mirror their approach and precompute
    $\lnMom{p}{R}{1}$ at $p \in \{0.0, 0.1, \dotsc, 49.9, 50.0\}$ and
    $\lnMom{p}{R}{q}$ at $(p, q) \in \{1.0, 1.1, \dotsc, 9.9, 10.0\}^2$.
    Doing so is slow but necessary for an accurate comparison.  Speed
    optimizations are left for future work.

\begin{figure}
    \begin{tikzpicture}
        \begin{axis} [
            xmode=log, xlabel={actual count (geometric mean of)},
            ymode=log, ylabel={upper bound (geometric mean of)},
            graph marker
        ]
            \fill [yellow] (1, 1) -- (1e20, 1) -- (1e20, 1e20);
            \addplot table [x=true, y=dex, meta=shape] {fig37/average.csv};
            \addplot table [x=true, y=ambi, meta=shape] {fig37/average.csv};
        \end{axis}
    \end{tikzpicture}
    \Description{
        A comparison of upper bounds to actual counts.
    }
    \caption{
        The geometric mean of Figures~\ref{fig:plots1}--\ref{fig:plots3},
        i.e., averaging over almost all undirected datasets on SNAP.  The
        horizontal axis is the actual count of homomorphisms; the vertical
        axis is for upper bounds.  Taking geometric means avoids large
        datasets overshadowing small ones.  The region below the diagonal
        is shaded in yellow to help compare relative errors.
    }                                                     \label{fig:average}
\end{figure}
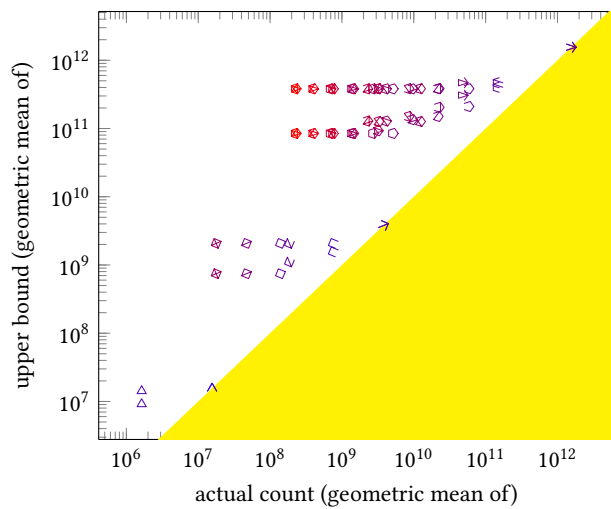

\def\presentdataset#1{%
    \begin{tikzpicture}
        \begin{axis} [
            width=6cm, height=5cm, ticklabel style={overlay, inner sep=1pt},
            xmode=log, ymode=log, try min ticks log=4,
            graph marker % max space between ticks=30,
        ]
            \fill [yellow] (1, 1) -- (1e20, 1) -- (1e20, 1e20);
            \draw (rel axis cs: 0.9, 0.1) node[above left]{\texttt{#1}};
            \addplot table [x=true, y=dex, meta=shape] {fig37/#1.csv};
            \addplot table [x=true, y=ambi, meta=shape] {fig37/#1.csv};
        \end{axis}
        \path (current bounding box.south west) + (-0.5cm, -0.5cm);
    \end{tikzpicture}%
    \hfill
}

\begin{figure}
    \catcode`\_=11
    \leavevmode
    \presentdataset{artist_edges}
    \presentdataset{as20000102}
    \presentdataset{athletes_edges}
    \presentdataset{Brightkite_edges}
    \presentdataset{CA-AstroPh}
    \presentdataset{CA-CondMat}
    \presentdataset{CA-GrQc}
    \presentdataset{CA-HepPh}
    \presentdataset{CA-HepTh}
    \presentdataset{com-amazon}
    \presentdataset{com-dblp}
    \presentdataset{com-youtube}
    \presentdataset{company_edges}
    \presentdataset{deezer_europe_edges}
    \presentdataset{Email-Enron}
    \unskip
    \Description{
        Detailed comparison of upper bounds to actual counts.
    }
    \caption{
        Upper bounds vs.\ actual counts plots for undirected SNAP datasets
        (A--E).
        }                                                  \label{fig:plots1}
\end{figure}

\begin{figure}
    \catcode`\_=11
    \leavevmode
    \presentdataset{facebook_combined}
    \presentdataset{government_edges}
    \presentdataset{Gowalla_edges}
    \presentdataset{HR_edges}
    \presentdataset{HU_edges}
    \presentdataset{lastfm_asia_edges}
    \presentdataset{musae_chameleon_edges}
    \presentdataset{musae_crocodile_edges}
    \presentdataset{musae_DE_edges}
    \presentdataset{musae_ENGB_edges}
    \presentdataset{musae_ES_edges}
    \presentdataset{musae_facebook_edges}
    \presentdataset{musae_FR_edges}
    \presentdataset{musae_git_edges}
    \presentdataset{musae_PTBR_edges}
    \unskip
    \Description{
        Detailed comparison of upper bounds to actual counts.
    }
    \caption{
        Upper bounds vs.\ actual counts plots for undirected SNAP datasets
        (F--M).
    }                                                      \label{fig:plots2}
\end{figure}

\begin{figure}
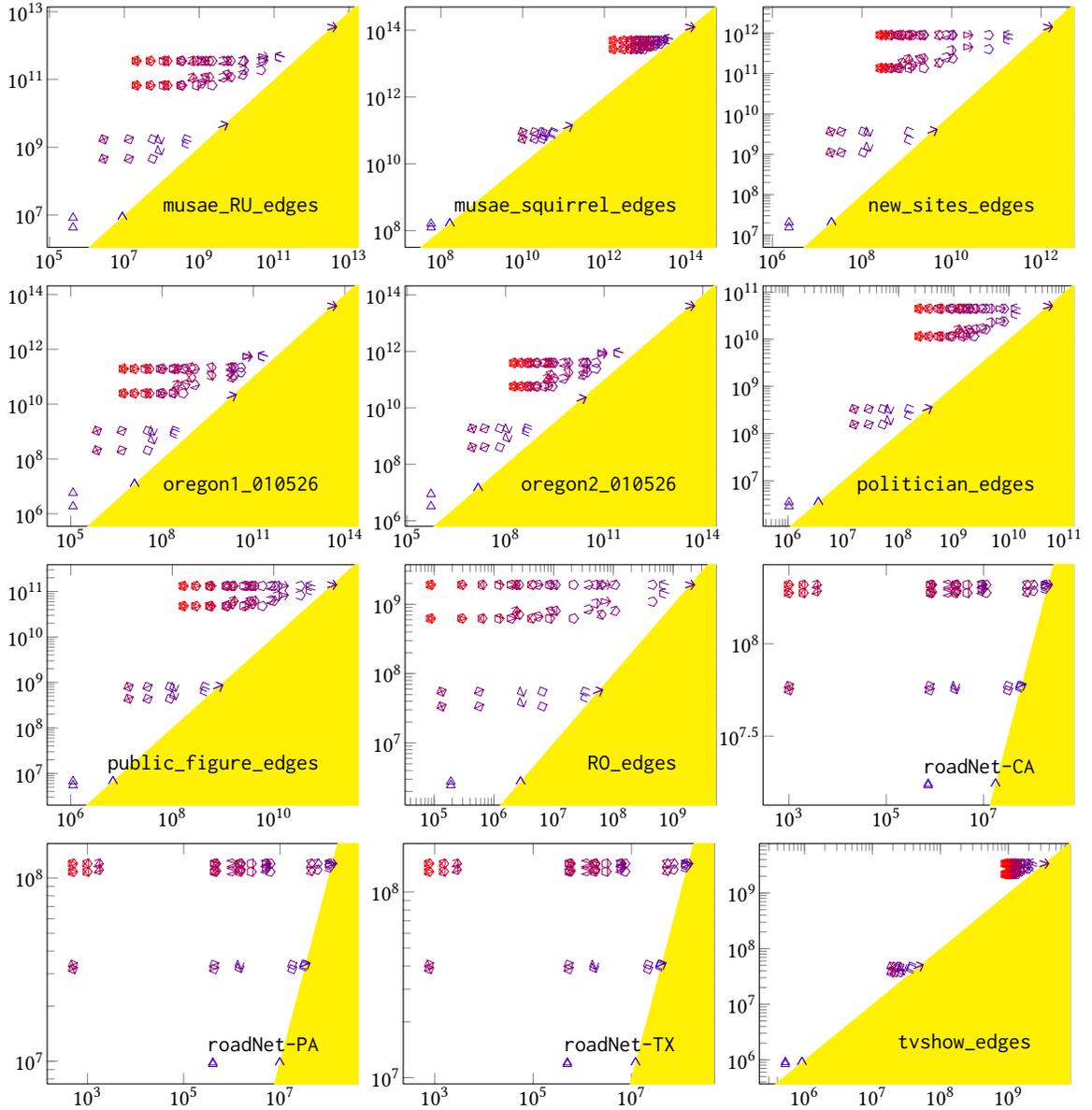

    \catcode`\_=11
    \leavevmode
    \presentdataset{musae_RU_edges}
    \presentdataset{musae_squirrel_edges}
    \presentdataset{new_sites_edges}
    \presentdataset{oregon1_010526}
    \presentdataset{oregon2_010526}
    \presentdataset{politician_edges}
    \presentdataset{public_figure_edges}
    \presentdataset{RO_edges}
    \presentdataset{roadNet-CA}
    \presentdataset{roadNet-PA}
    \presentdataset{roadNet-TX}
    \presentdataset{tvshow_edges}
    \unskip
    \Description{
        Detailed comparison of upper bounds to actual counts.
    }
    \newbox\Kfivebox\setbox\Kfivebox=\hbox{%
        \tikz[scale=2]{\pgfuseplotmark{K5}}%
    }
    \caption{
        Upper bounds vs.\ actual counts plots for undirected SNAP datasets
        (M--T).  Note that \texttt{roadNet-CA}, \texttt{roadNet-PA}, and
        \texttt{roadNet-TX} do not have any $K\sb5$ (no five roads
        intersecting each other), so the corresponding \usebox\Kfivebox\
        icons fall outside.
    }                                                      \label{fig:plots3}
\end{figure}

\begin{figure}
    \begin{tikzpicture}
        \begin{axis} [
            xlabel={$\log_{10}(\text{dexterous}/\text{actual})$ (mean of)},
            ylabel={$\log_{10}(\text{ambidextrous}/\text{actual})$ (mean of)},
            graph marker
        ]
            \draw [yellow, line width=2pt] (-100, -74.81) -- (100, 74.81);
            \addplot table [x=dex_rel, y=ambi_rel, meta=shape]
                {fig37/average.csv};
        \end{axis}
    \end{tikzpicture}
    \Description{
        A comparison of relative errors.
    }
    \caption{
        Averaged log of relative errors.  The origin-passing least-square
        line is drawn in yellow.  Its slope is $0.7481$; its $R^2$ is
        $0.9870$.  This means that, when dexterous overestimates by $x$-fold,
        ambidextrous overestimates by $x^{0.7481}$-fold.
    }                                                    \label{fig:relative}
\end{figure}
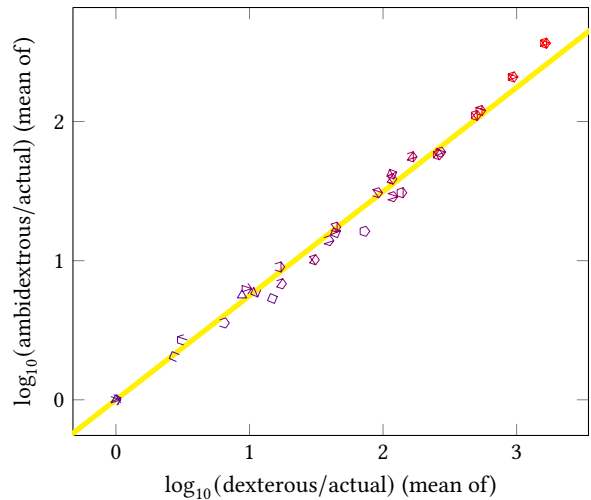

\def\presentdataset#1{%
    \begin{tikzpicture}
        \begin{axis} [
            width=6cm, height=5cm, ticklabel style={overlay, inner sep=1pt},
            graph marker
        ]
            \draw (rel axis cs: 0.9, 0.1) node[above left]{\texttt{#1}};
            \addplot table [x=dex_rel, y=ambi_rel, meta=shape]
                {fig37/#1.csv};
        \end{axis}
        \path (current bounding box.south west) + (-0.5cm, -0.5cm);
    \end{tikzpicture}%
    \hfill
}

\begin{figure}
    \catcode`\_=11
    \leavevmode
    \presentdataset{artist_edges}
    \presentdataset{as20000102}
    \presentdataset{athletes_edges}
    \presentdataset{Brightkite_edges}
    \presentdataset{CA-AstroPh}
    \presentdataset{CA-CondMat}
    \presentdataset{CA-GrQc}
    \presentdataset{CA-HepPh}
    \presentdataset{CA-HepTh}
    \presentdataset{com-amazon}
    \presentdataset{com-dblp}
    \presentdataset{com-youtube}
    \presentdataset{company_edges}
    \presentdataset{deezer_europe_edges}
    \presentdataset{Email-Enron}
    \unskip
    \Description{
        Detailed comparison of upper bounds to actual counts.
    }
    \caption{
        Relative error plots for undirected SNAP datasets (A--E).
    }                                                      \label{fig:plots4}
\end{figure}

\begin{figure}
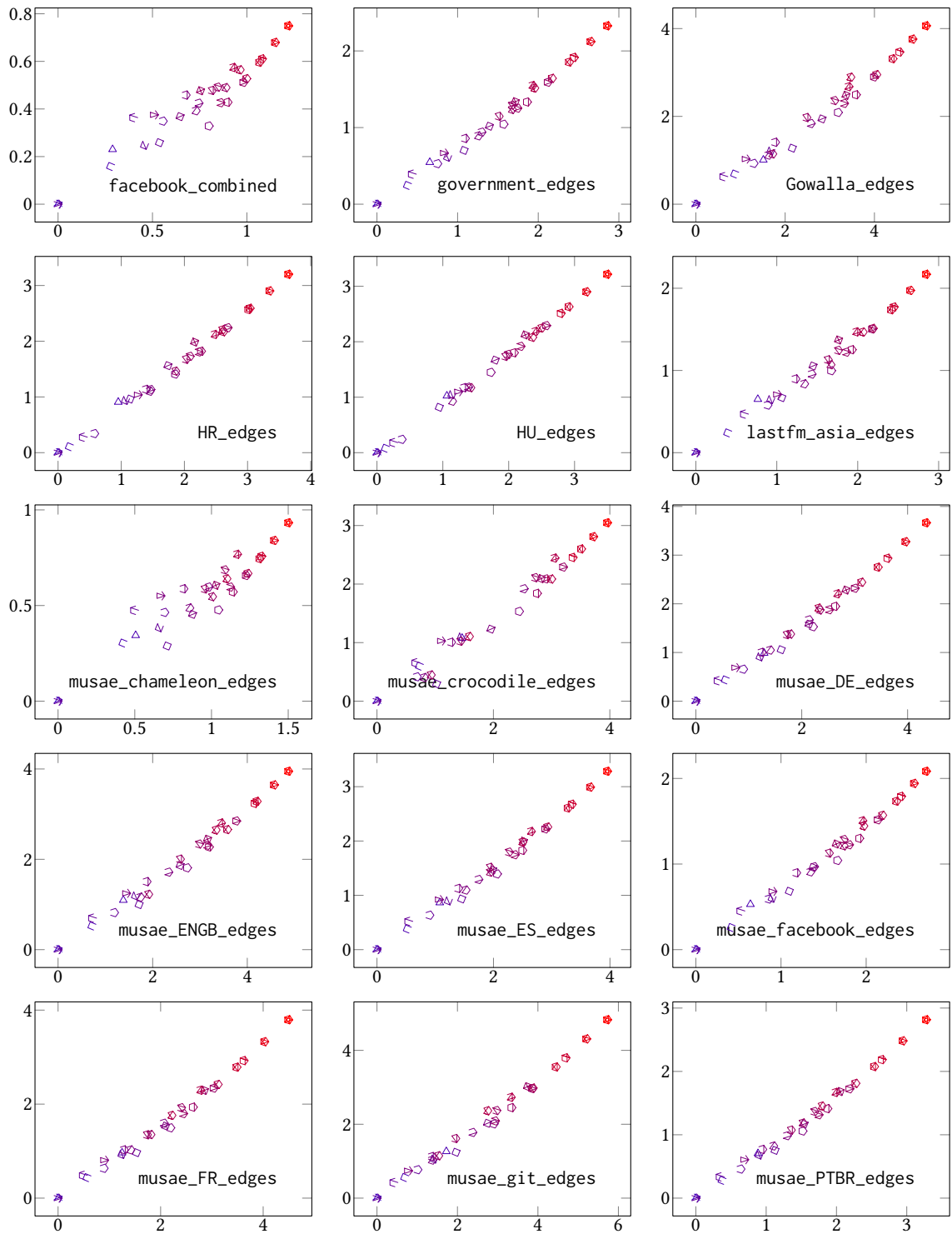

    \catcode`\_=11
    \leavevmode
    \presentdataset{facebook_combined}
    \presentdataset{government_edges}
    \presentdataset{Gowalla_edges}
    \presentdataset{HR_edges}
    \presentdataset{HU_edges}
    \presentdataset{lastfm_asia_edges}
    \presentdataset{musae_chameleon_edges}
    \presentdataset{musae_crocodile_edges}
    \presentdataset{musae_DE_edges}
    \presentdataset{musae_ENGB_edges}
    \presentdataset{musae_ES_edges}
    \presentdataset{musae_facebook_edges}
    \presentdataset{musae_FR_edges}
    \presentdataset{musae_git_edges}
    \presentdataset{musae_PTBR_edges}
    \unskip
    \Description{
        Detailed comparison of upper bounds to actual counts.
    }
    \caption{
        Relative error plots for undirected SNAP datasets (F--M).
    }                                                      \label{fig:plots5}
\end{figure}

\begin{figure}
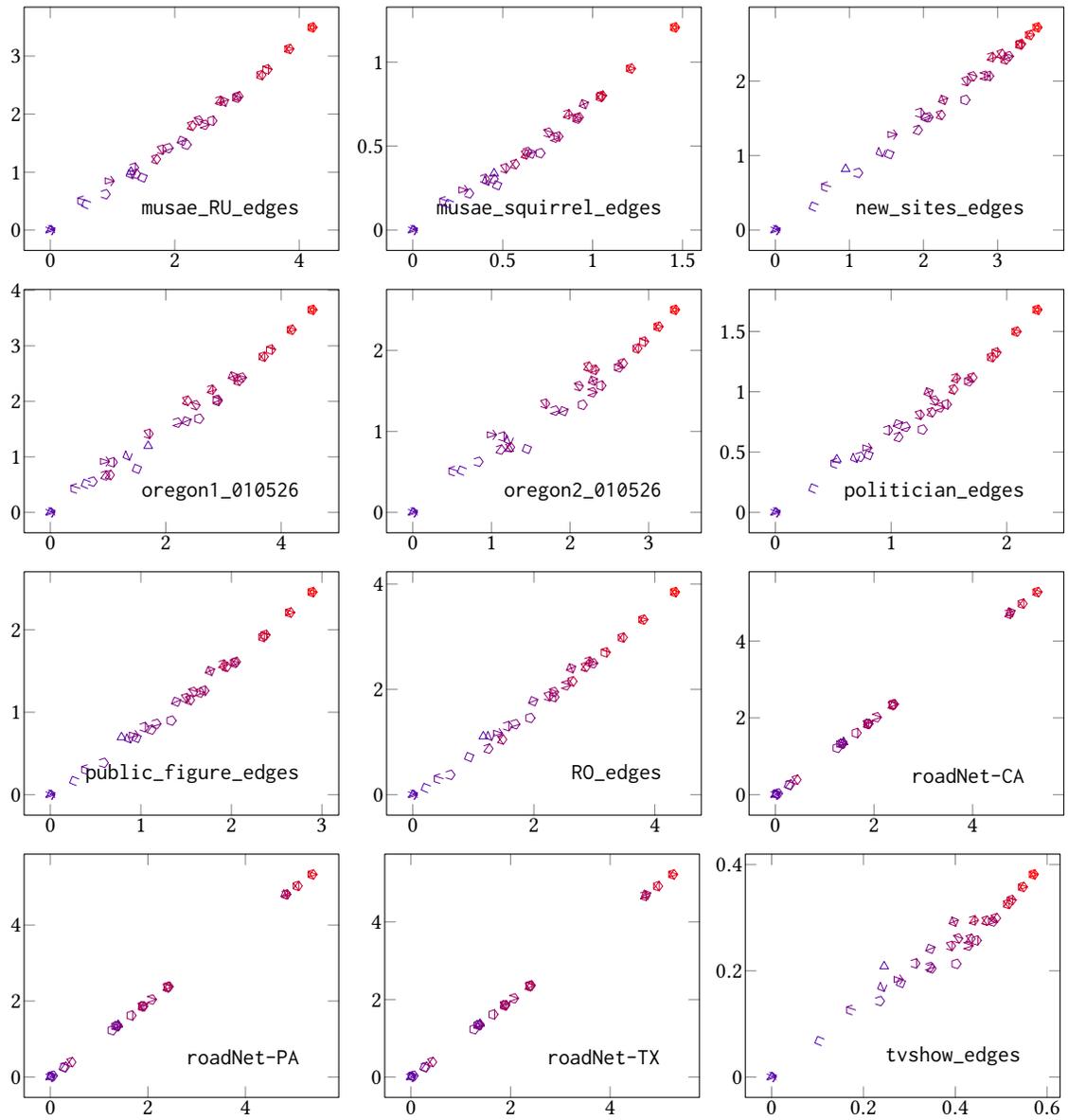

    \catcode`\_=11
    \leavevmode
    \presentdataset{musae_RU_edges}
    \presentdataset{musae_squirrel_edges}
    \presentdataset{new_sites_edges}
    \presentdataset{oregon1_010526}
    \presentdataset{oregon2_010526}
    \presentdataset{politician_edges}
    \presentdataset{public_figure_edges}
    \presentdataset{RO_edges}
    \presentdataset{roadNet-CA}
    \presentdataset{roadNet-PA}
    \presentdataset{roadNet-TX}
    \presentdataset{tvshow_edges}
    \unskip
    \Description{
        Detailed comparison of upper bounds to actual counts.
    }
    \newbox\Kfivebox\setbox\Kfivebox=\hbox{%
        \tikz[scale=2]{\pgfuseplotmark{K5}}%
    }
    \caption{
        Relative error plots for undirected SNAP datasets (M--T).  Note that
        \texttt{roadNet-CA}, \texttt{roadNet-PA}, and \texttt{roadNet-TX} do
        not have any $K\sb5$ (no five roads intersecting each other), so the
        corresponding \usebox\Kfivebox\ icons fall outside.
    }                                                      \label{fig:plots6}
\end{figure}

\section{Conclusions}

    In the present paper, we generalize degree sequences to bi-degree
    sequences, $p$-norms to bivariate moments, and dexterous bounds to
    ambidextrous bounds.  Our new bounds yield pessimistic cardinality
    estimates that are tighter, as confirmed by both inequalities and
    simulations.  We also prove that both old and new bounds are convex
    functions, which allows convex optimization techniques to step in.
    Future work might include using sketching and spline to strike a balance
    between speed and accuracy of bivariate moments.

\bibliographystyle{ACM-Reference-Format}
\bibliography{AmbidexterCardinal-29}

\end{document}